\documentclass[journal,onecolumn]{IEEEtran}

\usepackage{graphicx,mathptmx,amsmath,amsfonts,amsthm,color,hyperref,comment}
\usepackage{enumitem}

\newtheorem{assumption}{Assumption}
\newtheorem{theorem}{Theorem}
\newtheorem*{example*}{Example}

\newtheorem{proposition}{Proposition}

\def\ben{\begin{eqnarray*}}
\def\een{\end{eqnarray*}}

\def  \A{\mathbb {A}}

\def  \R{\mathbb {R}}

\def  \E{\mathbb {E}}

\def \cC{\mathcal{C}}

\def\GN{ G }
\def\HN{ H }
\def\Hmin{ H_{\min} }
\def\Hhalf{ H_{1/2} }
\def\LambdaN{ \Lambda^N }
\def\TURN{ \gamma }
\def\IN{ I^N }
\def\Noise{ N }

\def\IU{ I^U }
\def\IUN{ I^{U/N} }

\def\xstar{ x^* }
\def\success{ s }
\def\error{ \epsilon }
\def\errorAML{ \epsilon^\text{AB} }

\def\ystar{ y^* }

\def\INUD{ I^{N,U,D} }

\def\GML{ D }
\def\IMLE{I^\text{GRAND}}

\def\GAML{ D_\text{AB} }

\def\pabandon{ p_\text{abandon} }
\def\pblock{ p_\text{block} }

\begin{document}

%---------------------------------------------------------------------------------------------------------------------------

\title{Capacity-achieving Guessing Random Additive Noise Decoding (GRAND)}

\author{
\IEEEauthorblockN{Ken R. Duffy\IEEEauthorrefmark{1}, Jiange Li\IEEEauthorrefmark{2} and Muriel M\'edard\IEEEauthorrefmark{2}\\}
\IEEEauthorblockA{\IEEEauthorrefmark{1}Hamilton Institute, Maynooth University, Ireland. E-mail: ken.duffy@mu.ie.\\}
\IEEEauthorblockA{\IEEEauthorrefmark{2}Research Laboratory of Electronics, Massachusetts Institute of Technology, Cambridge, MA 02139, U. S. A.
E-mail: lijiange@mit.edu, medard@mit.edu.}
}

\maketitle

%------------------------------------------------------------------------------------------------------------------------------------

\begin{abstract}
We introduce a new algorithm for realizing Maximum Likelihood (ML)
decoding in discrete channels with or without memory. In it, the
receiver rank orders noise sequences from most likely to least
likely.  Subtracting noise from the received signal in that order,
the first instance that results in a member of the code-book is the
ML decoding. We name this algorithm GRAND for Guessing Random
Additive Noise Decoding.

We establish that GRAND is capacity-achieving when used with random
code-books. For rates below capacity we identify error exponents,
and for rates beyond capacity we identify success exponents. We
determine the scheme's complexity in terms of the number of
computations the receiver performs. For rates beyond capacity, this
reveals thresholds for the number of guesses by which if a member
of the code-book is identified it is likely to be the transmitted
code-word.

We introduce an approximate ML decoding scheme where the receiver
abandons the search after a fixed number of queries,  an approach
we dub GRANDAB, for GRAND with ABandonment. While not an ML decoder,
we establish that the algorithm GRANDAB is also capacity-achieving
for an appropriate choice of abandonment threshold, and characterize
its complexity, error and success exponents. Worked examples are
presented for Markovian noise that indicate these decoding schemes
substantially out-perform the brute force decoding approach.

\end{abstract}


%------------------------------------------------------------------------------------------------------------------------------------

\section{Introduction}

\let\thefootnote\relax\footnote{\small\it 
These results were presented in part at ITA 2018, and in part at the 2018
International Symposium on Information Theory, Colorado, USA.}

Consider a discrete channel with inputs, $X^n$, and outputs, $Y^n$, consisting
of blocks of $n$ symbols from a finite alphabet $\A$ of size
$|\A|$. Assume that channel input is altered by random, not necessarily
memoryless, noise, $\Noise^n$, that is independent of the channel
input and also takes values in $\A^n$. Assume that the function,
$\oplus$, describing the channel's action,
\begin{align}
    \label{eq:channel}
    Y^n = X^n \oplus \Noise^n,
\end{align}
is invertible, so that knowing the output and input the noise can be recovered:
\begin{align}
    \label{eq:channel_invert}
    X^n=Y^n \ominus \Noise^n.
\end{align}
To implement Maximum-Likelihood (ML) decoding,
the sender and receiver first share a code-book
$\cC_n=\{c^{n,1},\ldots,c^{n,M_n}\}$ consisting of $M_n$ elements
of $\A^n$. For a given channel output $y^n$, 
denote the conditional probability of the received sequence for
each code-word in the code-book by
\begin{align}
    \label{eq:cond}
    p(y^n|c^{n,i}) = P(y^n=c^{n,i}\oplus \Noise^n) \text{ for } i\in\{1,\ldots,M_n\}.
\end{align}
The decoding produced by GRAND is then an element of the code-book that has the
highest conditional likelihood of transmission given what was
received,
\begin{align}
    \label{eq:straight_MLE}
    c^{n,*}\in\arg\max \left\{ p(y^n|c^{n,i}): c^{n,i}\in\cC_n\right\}
    =\arg\max \left\{ P(\Noise^n=y^n\ominus c^{n,i}):c^{n,i}\in\cC_n\right\},
\end{align}
where we have used the invertibility of $\oplus$ for the final equality.

Code-book sizes are typically exponential in the block length $n$,
$M_n \sim |\A|^{nR}$ and, taking logs throughout the article as
base $|\A|$, we define the normalized rate of the code-book to be
$R=\lim_n 1/n \log(M_n)$. Thus ML decoding would appear to be
infeasible in practice for reasonable rates as it would seem that
the receiver would have to either: A) perform $|\A|^{nR}$ conditional
probability computations described in equation \eqref{eq:cond}
followed by a rank ordering every time a signal is received; or B),
in advance, perform $|\A|^{n(R+1)}$ computations described in
equation \eqref{eq:cond}, one for every $(c^{n,i},y^n)$ pair, storing
in a look-up table the resulting $|\A|^n$ ML decodings, one for
each possible received sequence.

In the present paper we consider a distinct algorithm for ML decoding.
The principle underlying the approach is for the receiver to
rank-order noise sequences from most likely to least likely and
then sequentially query whether the sequence that remains when the
noise is removed from the received signal is an element of the
code-book.  For the channel structure described above, irrespective
of how the code-book is constructed, the first
instance where the answer is in the affirmative corresponds to the
ML decoding. More formally, the receiver first creates an ordered
list of noise sequences, $\GN:\A^n\mapsto\{1,\ldots,|\A|^n\}$, from
most likely to least likely, breaking ties arbitrarily:
\begin{align}
\label{eq:noise_order}
\GN(z^{n,i})\leq \GN(z^{n,j}) \text{ if and only if } P(\Noise^n=z^{n,i})\geq P(\Noise^n=z^{n,j}),
\end{align}
where throughout this article lower case letters correspond to
realizations of upper-case random variables, apart from for noise
where $z$ is used as $n$ denotes the code block-length.
For each received signal, the receiver executes the following algorithm, which we call GRAND for Guessing Random Additive Noised Decoding:
\begin{itemize}[leftmargin=*]
\item Given channel output $y^n$, initialize $i=1$ and set $z^n$
to be the most likely noise sequence, i.e. the $z^n$ such that
$\GN(z^n)=i$.
\item While $x^n=y^n\ominus z^n\notin\cC_n$, increase $i$ by $1$
and set $z^n$ to be the next most likely noise sequence, i.e. the
$z^n$ such that $\GN(z^n)=i$.
\item The $x^n$ that results from this while loop is the decoded
element.
\end{itemize}
An example of this process is described in Table \ref{table:inner_code}. 

To see that GRAND corresponds to ML decoding for channels
of the sort described in equations \eqref{eq:channel} and
\eqref{eq:channel_invert}, note that, owing to the definition of
$c^{n,*}$ in equation \eqref{eq:straight_MLE},
\begin{align*}
    P(\Noise^n=y^n\ominus c^{n,*})\geq P(\Noise^n=y^n\ominus c^{n,i}) \text{ for all } c^{n,i}\in\cC_n.
\end{align*}
Thus the scheme does, indeed, identify an ML decoding. The premise
of the present paper is that there are circumstances when this new
scheme, GRAND, has a complexity that decreases as code-book rate increases even though the more direct approach
described in equation \eqref{eq:straight_MLE} sees steeply increasing complexity. 

\begin{table*}[t]
\begin{center}
\begin{tabular}{|c||c|c|c|c|c|c|c|}
\hline
 Noise guessing order & 
 1 & 2 & 3 & 4 & 5 & 6 & $\ldots$  \\ 
 \hline
 Noise from most likely to least likely & 
 $z^{n,1}$ & $z^{n,2}$ & $z^{n,3}$ & $z^{n,4}$ & $z^{n,5}$ &$z^{n,6}$ & $\ldots$ \\  
  \hline
 String queried for membership of the code-book $\cC_n$
 & $y^n\ominus z^{n,1}$ & $y^n\ominus z^{n,2}$ &  $y^n\ominus z^{n,3}$ & $y^n\ominus z^{n,4}$ & $y^n\ominus z^{n,5}$ & $y^n\ominus z^{n,6}$  & $\ldots$ \\
 \hline
 Location of code-book elements &
 & & $c^{n,i_1}$ & & $c^{n,i_2}$ & & $\ldots$ \\
 \hline
\end{tabular}
\end{center}
\caption{Description of ML decoding by  GRAND. The receiver
creates a rank-ordered list of noise from most likely to least
likely breaking ties arbitrarily, $z^{n,1},z^{n,2},\ldots$. In that
order, given a received signal $y^n$, the receiver sequentially
subtracts the noise $z^{n,i}$ and queries if the string that results,
$y^n\ominus z^{n,i}$, is an element of the code-book, $\cC_n$. The
first string that is in the code-book, is the ML decoding. In this
example, $c^{n,i_1}$ is the first element of the code-book to be
identified, which occurs on the third noise guess. In Approximate
ML decoding, GRANDAB, after a fixed number of queries the receiver abandons 
the questioning and declares an error.}
\label{table:inner_code}
\end{table*}

In Section \ref{subsec:analysis}, the performance of the algorithm
is established in terms of its maximum achievable rate, which is
a property of ML decoding rather than being particular to the present
GRAND scheme, and the number of computations the receiver must perform
until decoding, which is dependent on the scheme. With some mild
ergodicity conditions imposed on the noise process, we prove that
GRAND is capacity achieving with uniform-at-random code-books.
We determine asymptotic error exponents, as well as providing success
exponents for rates above capacity. We identify the asymptotic
complexity of GRAND in terms of the number of local operations
the receiver must perform per received block in executing the
algorithm.

Based on this new noise-centric design ethos for ML decoding and
the intuition that comes from its analysis, we introduce a new
approximate ML decoder in Section \ref{subsec:AML}, an approach we dub GRANDAB for GRAND with ABandonment. In this
variant of GRAND, the receiver abandons identification of the transmitted
code word if no element of the code-book is identified after a
pre-defined number of noise removal queries. GRANDAB is not a ML decoder
as the algorithm sometimes terminates without returning an element
of the code-book. Despite that, we establish that GRANDAB is also
capacity achieving for random code-books once the abandonment threshold
is set for after all elements of the Shannon Typical Set of
the noise are queried, and we determine the exponent for the
likelihood of abandonment. By abandoning after a fixed number of
queries, an upper-bound on complexity is ensured.

To determine these algorithmic properties, we leverage recent results
in the study of guesswork.  We recall one theorem from the literature
and establish several new ones. As they may appear somewhat
mathematically involved, we begin by explaining the intuitive meaning
behind them.

Theorem \ref{theorem:LDPN} is taken from \cite{Christiansen13} and
provides a Large Deviation Principle (LDP) as the block length,
$n$, becomes large, for the distribution of the logarithm of the
number of guesses needed until the actual noise in the channel is
queried, $\GN(\Noise^n)$. On its own, this result provides us with
an upper-bound on the complexity of the scheme, but it can be
augmented in the case of uniformly selected code-books. That is,
where the input elements $X^n$ in equation \eqref{eq:channel} are
chosen uniformly at random from a code-book $\cC_n$ that itself
consists of a collection of uniformly selected elements of $\A^n$. 

Theorem \ref{theorem:LDPU} is new and
establishes properties of the number of guesses that would be made
until an element of the code-book that was not the channel input
is identified. Here we leverage the fact that for uniformly distributed
code-books the location of each of these elements in the guessing
order outlined in Table \ref{table:inner_code} are uniform in
$\{1,\ldots,|\A|^n\}$. As a result, the distribution of the number
of guesses until any non-input element of the code-book is hit upon
is distributed as the minimum of $M_n$ such uniform random variables.
When $M_n \approx |\A|^{nR}$ and $n$ becomes large, the resulting
minimum is essentially the discretization of an exponential
distribution with rate $|\A|^{-n(1-R)}$ so that the receiver will
identify a code-word in, on average, approximately $|\A|^{n(1-R)}$
guesses.  Note, in particular, that as $R$ increases and the code-book
becomes more dense and efficient, while the number of computations
in the brute-force approach to ML decoding increases, the noise
guessing approach experiences the reverse phenomenon.

The ML decoding algorithm introduced in the present paper is essentially a
race between these two guessing processes. If the number of guesses
required to identify the true noise is less than the number of
guesses to identify any other element of the code-book, then GRAND provides the correct answer on termination. Combining
the two earlier results in two different ways first recovers the
Channel Coding Theorem as Proposition \ref{prop:channel_coding_theorem}
via this new guessing argument. Namely, with $R$ being the normalized
code-book rate, $\HN$ being the normalized Shannon entropy rate of
the noise base $|\A|$, and with $1-\HN$ being the channel capacity,
so long as $R<1-\HN$ then the ML decoder will correctly identify
the input for long enough blocks. The guessing argument provides
asymptotic exponents for the probability that the ML decoding
is an error if the code-book is within capacity, as well as for the
probability that the ML decoding is correct if the code-book rate
is beyond capacity. Both of these error and success exponents are
convex functions of the code-book rate near capacity and approach
zero at capacity, hinting at smooth degradation in performance near
capacity.

Combining Theorems \ref{theorem:LDPN} and \ref{theorem:LDPU} in a
distinct fashion akin to that used in \cite{Christiansen15} to study
multi-user guesswork, Proposition \ref{prop:IMLE} characterizes the
complexity of the scheme in terms of the distribution of the number
of guesses to termination. This approach allows us to determine
some subtle performance features of the scheme when code-books rates
are beyond capacity. Theorem \ref{theorem:WLLN} establishes that
the circumstances beyond capacity under which the ML decoding is
likely to be correct decoding should the noise guessing complete
quickly. In particular, this phenomenon occurs if the code-book
rate is less than one minus the min-entropy rate of the noise.

Interpreting the results of Propositions \ref{prop:channel_coding_theorem}
and \ref{prop:IMLE} in light of the noise guessing algorithm leads
us to propose a new approximate ML decoder, GRANDAB. In GRANDAB, if no
code-book element is identified by the receiver after
$|\A|^{n(\HN+\delta)}$ queries for some $\delta>0$, the receiver
abandons guessing and decoding results in an error. While it is
not an ML decoder, we prove in Proposition \ref{prop:AML} that GRANDAB
is also capacity-achieving for any $\delta>0$. Thus GRANDAB has the
capacity achieving qualities of ML decoding with a guaranteed upper
bound on the number of computations performed by the receiver. This
can result in a significant saving over GRAND in terms of complexity as
the average number of queries required to identify the true noise
in the system grows with an exponent of R\'enyi entropy rate
$1/2$.

In Section \ref{sec:examples} the performance of GRAND and GRANDAB are  illustrated for bursty Markovian channel noise
as, crucially, all of the results in this paper hold for channels
with memory, a point we investigate  in Section \ref{sec:conc}. For
memoryless channels, however, the guessing approach enables finer
approximations to the computation of block error probabilities than
asymptotic exponents and these are used for the BSC in Section
\ref{subsec:finer}. In Section \ref{sec:conc} we conclude with a
discussion of implementation and further potential of the principles
underlying the decoding algorithms introduced here.

\section{Related work}

Large deviation style arguments that are employed to establish error
exponents in both source and channel coding are typically variants
of Sanov's Theorem \cite{Dembo98} and the method of types. If sources are assumed
to have properties such as being independently and identically
distributed (IID) or Markovian, identification of non-asymptotic
pre-factors can be possible. For error exponents in source coding,
these methods have been used extensively, originally for asymptotically
error-free source coding with IID and Markov sources \cite{Ana90,
CK81, DLS81}, and then for variable-length and lossy source coding
of IID and stationary sources \cite{DK08}. For channel coding of
Discrete Memoryless Channels (DMCs), error exponents were first
identified by Gallager \cite{Gallager65} by direct arguments.
In unpublished notes that are available on the web, Montanari and
Forney \cite{MF} provide a relationship between Gallager's error
exponent and the exponent obtained through large deviation
considerations of channel coding arguments using asymptotic
equipartition principles. More recently, an approach along these
lines has been used to study joint source-channel coding \cite{YTWH17}.
As an aside, we remark that an alternate means of establishing the
results in \cite{YTWH17} would have been be to combine the results
of \cite{DK08} with the generalization in \cite{ZAC06} of \cite{Csi80,
Csi82} using a method of types.

While the arguments used in the papers referenced above are essentially
based on variants and refinements of the Large Deviation Principle
(LDP) for empirical measures, we instead analyze our proposed
approach starting from a completely distinct angle: the recently
established LDP \cite{Christiansen13} for Massey's guesswork
\cite{Massey94}. That LDP is a development from earlier results
that identify scaling exponents for moments of guesswork in terms
of R\'enyi entropy rates \cite{Arikan96, Malone04, Pfister04}.
Given the explicit relationship between the guesswork process and
the noise guessing approach, this seems the most natural line of
attack. In \cite{Ari2000} Arikan establishes LDP bounds for conditional probability rank. The full large deviation principle, which we employ here, is proven in \cite{Christiansen13}.

The connection between source coding and guesswork was first noted
by Arikan and Merhav \cite{AM98}, and has been established by Hanawal
and Sundaresan \cite{Hanawal11b}. For channel coding, a connection
between guesswork and error exponent analysis was proved by
Arikan for sequential decoding of tree codes \cite{Arikan96}, such as classic
convolutional codes \cite{Elias54}. Sequential decoding, introduced
by Wozencraft \cite{Wozencraft57, WR61}, is a variant of ML decoding
for tree codes. To ensure low computational complexity of sequential
decoding of convolutional codes, rates are generally kept
below a computational cutoff rate \cite{Wozencraft57, Fano1963,
Jacobs1967, Falconer69, Jelinek69, Hashimoto79, Arikan88, Arikan96}.
A survey of the historical rationale for cut-off rate design can
be found in \cite{NS94}. Several schemes have sought to improve
the cut-off rate, including Pinsker's concatenated code with an
inner block code and outer sequential decoder \cite{Pinsker65}, as
well as Massey's ``splitting" argument for quaternary erasure channel
\cite{Massey81}. A general framework for designing codes that
increase the cutoff rate is discussed in \cite{Arikan16}. Polar
Coding, which is capacity achieving for binary DMCs \cite{Arikan09},
fits into that framework.

For linear block codes, an ML decoding method exists that has complexity bounded by $2^{n(1-R)}$ (in the current article's notation) \cite{BCJR74}.
As the complexity of brute force ML decoding is $2^{nR}$, that approach is preferable when $1-R<R$, that is when $R>1/2$. For rates below capacity, $R<1-H$ and hence $H<1-R$. GRANDAB's complexity $2^{n(H+\delta)}$, for arbitrary $\delta>0$, is thus lower than the one provided by \cite{BCJR74}, except for low code rates where the complexity of brute force ML decoding is preferable. The approach taken in \cite{BCJR74} is based on a trellis decoding method for linear convolutional codes akin to the one independently derived in \cite{McA74}, in which terminated, or so called blocked, convolutional codes are also considered.

In the present work, we do not envisage designing codes, but using
random ones. For the channels we are considering, Shannon's
\cite{Shannon48} uniform random code-book plus ML decoding argument
affords capacity, but for codes of sufficient length that approach
capacity, decoding methods for random codes are prohibitively complex
with existing methods, as explained in the introduction. The core
performance idea here is to leverage the fact that the noise is
typically highly non-uniform, rendering its identification through
guessing less onerous than performing a computation for every element
of the code-book.

While our model employs uniformly distributed code-words, we analyze
substantially more general noise processes than the DMC. For the
DMC, the error exponent we derive necessarily matches Gallager's.
That is unsurprising, as he proves it was tight for the average code
\cite{Gal73}, and this fact has recently generalized to random
linear codes \cite{DZF16} for channels for which uniform code-books
are optimal. As an aside, we remark that the result in \cite{Gal73}
is echoed in the source coding domain in \cite{DK08}, which shows,
via asymptotic equipartition style arguments, that almost all random
code-books provide in effect the same compression performance. Thus,
one might suspect that results analogous to those in \cite{DZF16}
are likely to hold also for source coding \cite{Csi82b} and network
coding \cite{Hoetal06, EMHRKKH03}.

The mathematical approach we take naturally lends itself to the
determination of decay exponents in the probability of success when
coding above capacity. The question of success for codes operating
above capacity is a long-standing, though perhaps less well studied
than that of errors below capacity \cite{Wol61, Gal68, Ebe66}. For
a DMC, lower bounds \cite{Ari73} that are coincident with upper
bounds \cite{DK79} are known to exist. Here, the derivation of these
exponents come hand-in-hand with the determination of error exponents,
and hold for the same broad class of noise processes. 

GRAND employs ordered statistics of noise for decoding, but the
code-book is only used when checking if a proposed decoded code
word pertains to the code-book. The noise statistics may be obtained
by arbitrary means and are not dependent on examining the decoder's
output. This approach differs from Ordered Statistics Decoding (OSD)
\cite{FL95, KI02}, which uses the statistics derived from syndrome
computations to update soft information in decoding linear bock
codes, or from Turbo-style systems that blend decoding with soft
information, see for instance \cite{HP94, KImetal10, WP99, SSS04}.

As ML decoding is generally too onerous from a complexity perspective,
the use of approximate ML decoding is, under different guises,
almost omnipresent in decoding algorithms. The approach GRANDAB
takes, that of stopping after a given set of guesses, is redolent
of limited search approaches commonly used in the decoding of
convolutional codes, such as reduced state sequence estimation
(RSSE) and related techniques that limit the search space in
sequential decoding \cite{And89, Fos77, AM84,  Sim90, SS90, MRDP94,
EQ89, SS92}. This latter family of techniques uses the received
sequence as a starting point, rather than consider the noise itself
as we do in GRANDAB, and most have not been formally established
to be capacity achieving.

\section{Analysis}
\label{sec:asymptotic}

\subsection{ML decoding by guesswork}
\label{subsec:analysis}

We begin with the assumption we shall make on the noise process.
Recall that $\log$ is taken base $|\A|$ throughout.

Assuming it exists, define the R\'enyi entropy rate of the noise
$\{\Noise^n\}$ process with parameter $\alpha\in(0,1)\cup(1,\infty)$ to be
\begin{align*}
H_\alpha 
        = \lim_{n\to\infty} \frac 1n \frac{1}{1-\alpha}\log\left(\sum_{z^n\in\A^n} P(\Noise^n=z^n)^\alpha\right),
\end{align*}
with $\HN=H_1$ being the Shannon entropy rate of the noise. 
Denote the min-entropy rate of the noise by $\Hmin =
\lim_{\alpha\to\infty} H_\alpha$.

\begin{assumption}
\label{ass:N}

We assume that
\begin{align}
\label{eq:LambdaN}
\LambdaN(\alpha)
=\lim_{n\to\infty}\frac 1n \log \E(\GN(\Noise^n)^\alpha)
=
\begin{cases}  
	\displaystyle\alpha  H_{1/(1+\alpha)}
	& \text { for } \alpha\in(-1,\infty)\\
	-\Hmin & \text{ for } \alpha\leq-1,
\end{cases}
\end{align}
and that the derivative of $\LambdaN(\alpha)$ is continuous on the range
$\alpha\in(-1,\infty)$. 
\end{assumption}
Assumption \ref{ass:N} is known to be satisfied for a broad range
of sources including i.i.d. \cite{Arikan96}, Markovian \cite{Malone04},
a large class of general, stationary processes \cite{Pfister04} and
others \cite{Hanawal11}; the condition for $\alpha\leq-1$ is established
for all of these in \cite{Christiansen13}. 

Note that, by setting $\alpha=1$, as first identified by Arikan
\cite{Arikan96}, from equation \eqref{eq:LambdaN} one has that the
average number of guesses required to identify the true noise grows
exponentially in block size, $n$, with R\'enyi entropy rate at
parameter $1/2$, $\Hhalf$, which is no smaller than the Shannon
entropy rate, $\HN$, of the noise. Note also that $\LambdaN(\alpha)$
has a continuous derivative everywhere except potentially at
$\alpha=-1$.  An operational meaning to the discontinuous derivative
when evaluated from above is identified in \cite{Christiansen13},
where the value of the discontinuity captures the exponential growth
in $n$ of the size of the set of most-likely noise sequences.

\begin{figure*}
\begin{center}
\includegraphics[width=0.80\textwidth]{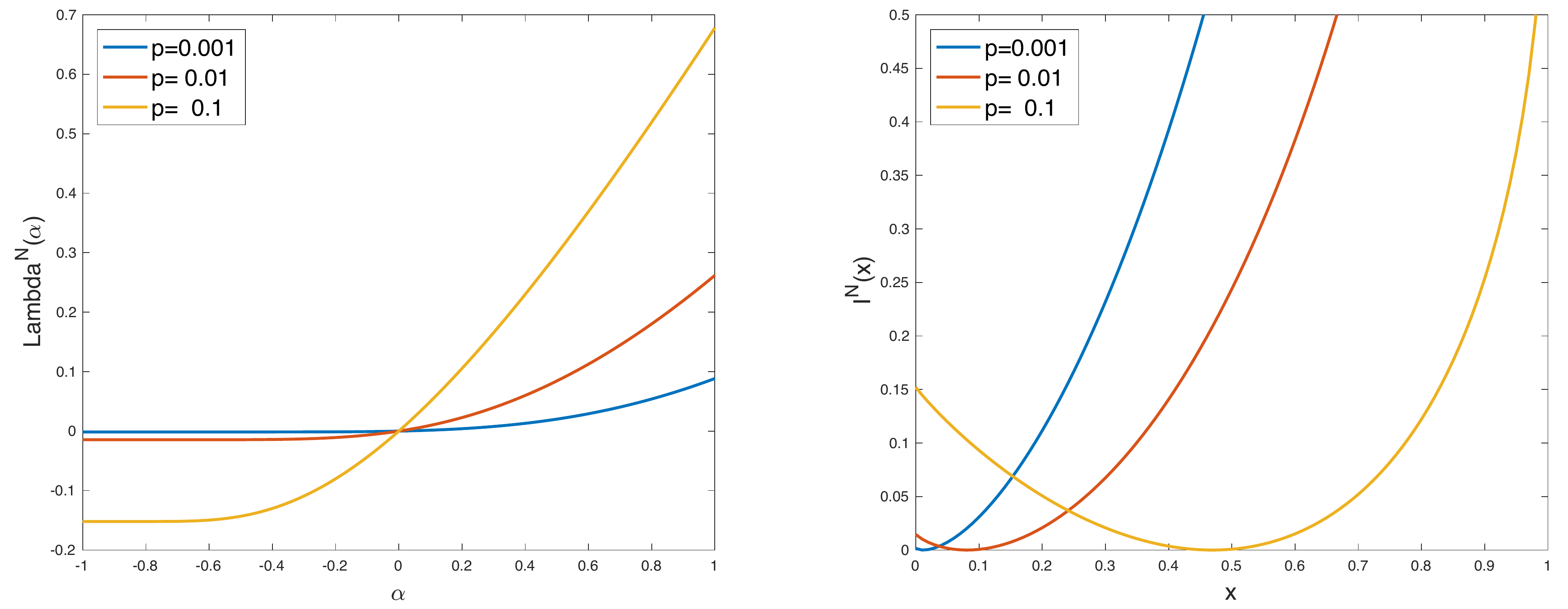}
\end{center}
\caption{Guesswork rate function. Example: $\A=\{0,1\}$, 
BSC channel,
noise $\Noise^n$ made of i.i.d. Bernoulli symbols with
$P(\Noise^1=1)=p\in(0,1)$. Left panel: scaled cumulant generating
function, $\LambdaN$, for the noise process $\{1/n\log\GN(\Noise^n)\}$
as determined by the explicit expression in \eqref{eq:BSC_K}. Right
panel: rate function for the same process, $\IN$, defined in equation
\eqref{eq:IK} and determined numerically. Roughly speaking $\log
P(n^{-1} \log\GN(\Noise^n) \approx x)\approx -n\IN(x)$.  Note that
$\IN(0)=\Hmin$, the min-entropy of the noise, and that the rate
function is zero at the Shannon entropy of the noise, $\IN(\HN)=0$.}
\label{fig:1}
\end{figure*}

\begin{example*}
For a BSC with $\A=\{0,1\}$ and an
additive channel mod 2, and $P(\Noise^1=1)=p$,
\begin{align}
\LambdaN(\alpha) = \begin{cases} 
	(1+\alpha)\log\left((1-p)^{\frac{1}{1+\alpha}} + p^{\frac{1}{1+\alpha}}\right) & \text{ if } \alpha\in(-1,\infty)\\
	\log(\max(1-p,p)) & \text{ if } \alpha \leq-1.
	\end{cases}
	\label{eq:BSC_K}
\end{align}
Plots of $\LambdaN(\alpha)$ can be found in Fig. \ref{fig:1}.
\end{example*}

From equation \eqref{eq:LambdaN}, $\LambdaN$ can be identified as
the scaled cumulant generating function for the process
$\{1/n\log\GN(\Noise^n)\}$ \cite{Dembo98} and so $\LambdaN$ is
necessarily convex. Moreover, that identification suggested that
this process may satisfy a Large Deviation Principle (LDP)
\cite{AM98, Sundaresan07}, which is proved in \cite{Christiansen13} and used in \cite{Christiansen13b, Beietal15, Chretal13, Beietalb15}.

\begin{theorem}[LDP for Guessing the Noise \cite{Christiansen13}]
\label{theorem:LDPN}
Under Assumption \ref{ass:N},
$\{1/n \log \GN(\Noise^n)\}$ satisfies the Large Deviation Principle with 
the convex lower-semi continuous rate function,
$\IN:[0,1]\to[0,\infty]$,
\begin{align}
    \IN(x) &:= \sup_{\alpha\in\R} \{x\alpha -\LambdaN(\alpha)\},
	\label{eq:IK}
\end{align}
which is the Legendre-Fenchel transform of $\LambdaN$.  

In particular: $\IN(0)=\Hmin$, the min-entropy rate of the noise;
$\IN(x)$ is linear on $[0,\TURN]$, where $\TURN:=\lim_{\alpha
\downarrow -1}d/d\alpha\,\LambdaN(\alpha)$, and then strictly convex
thereafter while finite; and $\IN(x)=0$ if and only if $x=\HN$, the
Shannon entropy rate of the noise.
\end{theorem}
Intuitively, this result says that,  for fixed  $(a, b)$, as $n$ goes to infinity
\begin{align*}
\log P\left(\frac1n \log \GN(\Noise^n) \in(a,b)\right) \approx -n\inf_{x\in(a,b)} \IN(x)
\end{align*}
for large $n$. As well as providing this approximation, one of the
primary advantages of a LDP over knowing how moments scale from
$\LambdaN$ is that it is covariant in the sense that LDPs are
preserved by continuous maps \cite{Dembo98}[Theorem 4.2.1], and we
shall repeatedly use that property to combine distinct LDPs.

\begin{example*}
While there is no closed form expression for $\IN$ for the BSC, it
can be readily computed numerically and examples are provided in
Fig. \ref{fig:1}.
\end{example*}

For random code-books, the second theorem provides a LDP for the
number of guesses on the noise that will be made until identifying
an element of the code-book that is not the transmitted code-word.
The key realization is that if elements of the code-book have been
selected uniformly at random, the location of the non-transmitted
code-book elements in the ordered list of noise guesses are also
uniform. 
Let $U^{n,1},\ldots,U^{n,M_n}$ be independent random
variables, each uniformly distributed in $\{1,\ldots,|\A|^n\}$ and 
define 
\begin{align*}
U^n =\min_i U^{n,i}.
\end{align*}

\begin{assumption}
\label{ass:U}
Assume that $M_n=\lfloor |\A|^{nR}\rfloor$ for some $R>0$.
\end{assumption}

\begin{theorem}[LDP for Guessing a Non-transmitted Code-word]
\label{theorem:LDPU}
Under Assumption \ref{ass:U}, as $n$ becomes large, $U^n$ 
is approximately exponentially distributed with rate $|\A|^{-n(1-R)}$,
\begin{align}
\label{eq:finerU}
\lim_{n\to\infty} P(|\A|^{-n(1-R)} U^n>x) = e^{-x} \text{ for all } x>0.
\end{align}
Moreover, $\{1/n \log U^{n}\}$ satisfies
the large deviation principle with lower semi-continuous rate function
\begin{align}
    \IU(x) = \begin{cases}
    1-R-x & \text{ if } x\in[0,1-R]\\
    +\infty & \text{ otherwise}
    \end{cases}
\label{eq:IU}
\end{align}
and 
\begin{align*}
\lim_{n\to\infty} \frac 1n \log \E(U^n) = 1-R.
\end{align*}
\end{theorem}
\begin{proof}
We begin by observing that
\begin{align*}
    P\left(U^{n}>|\A|^{xn}\right) = 
    \prod_{i=1}^{M_n} P\left(U^{n,i}> |\A|^{xn}\right)
	 =\left(1-\frac{\lceil|\A|^{xn}\rceil}{|\A|^n}\right)^{M_n}.
\end{align*}
Setting $x=1-R$ and making use of L'Hospital's rule, by assumption
\ref{ass:U} we have that for $y>0$
\begin{align*}
\lim_{n\to\infty} P(|\A|^{-n(1-R)} U^n>y) = 
\lim_{n\to\infty} \left(1-y|\A|^{-nR}\right)^{|\A|^{nR}}=e^{-y},
\end{align*}
giving equation \eqref{eq:finerU}.

As $[0,1]$ is compact, in order to establish the LDP it is sufficient
\cite{Dembo98} to prove that
\begin{align}
    \lim_{\epsilon\downarrow 0}\liminf_{n\to\infty} \frac1n \log P\left(\frac1n \log U^{n} \in (x-\epsilon,x+\epsilon)\right) 
    = \lim_{\epsilon\downarrow 0}\limsup_{n\to\infty} \frac1n \log P\left(\frac1n \log U^{n} \in (x-\epsilon,x+\epsilon)\right)
    =-\IU(x) 
	\label{eq:coincident}
\end{align}
for all $x\in[0,1]$.
Using the earlier observation, we have the following limiting
equality for the survival function
\begin{align*}
\lim_{n\to\infty} \frac 1n \log P\left(\frac 1n\log U^{n}>x\right) 
	&= 
	 \lim_{n\to\infty} \frac{M_n}{n} \log \left(1-\frac{\lceil|\A|^{xn}\rceil}{|\A|^n}\right)
	= 
	 \lim_{n\to\infty} \frac{|\A|^{nR}}{n} \log \left(1-|\A|^{n(x-1)}\right)\nonumber\\
	&=
    	 -\lim_{n\to\infty}  \frac1n |\A|^{n(R+x-1)}
	=
	\begin{cases}
    	0 & \text{ if } x\in[0,1-R]\\
    	-\infty & \text{ if } x\in(1-R,1).
	\end{cases}
\end{align*}
From this, we can confirm the veracity of equation \eqref{eq:coincident}
for all $x\in(1-R,1]$:
\begin{align*}
    \lim_{\epsilon\downarrow 0}\lim_{n\to\infty} \frac1n \log P\left(\frac1n \log U^{n} \in (x-\epsilon,x+\epsilon)\right) 
	\leq \lim_{\epsilon\downarrow 0}\lim_{n\to\infty} \frac1n \log P\left(\frac1n \log U^{n} > x-\epsilon\right) =-\infty.\\
\end{align*}
The corresponding equality for the cumulative distribution function can be
obtained by first noting that, by the Binomial theorem,
\begin{align*}
\lim_{n\to\infty} \frac{\left(1-|\A|^{n(x-1)}\right)^{|\A|^{nR}}}{1-|\A|^{n(R+x-1)}} = 1
\text{ if } x\in[0,1-R),
\end{align*}
while if $x=1-R$ the limit of the numerator in the above equation is $\exp(-1)$.
Thus to prove that equation \eqref{eq:coincident} holds for $x\in[0,1-R]$, we have
\begin{align*}
\lim_{\epsilon\downarrow 0}\lim_{n\to\infty} \frac1n \log P\left(\frac1n \log U^{n} \in (x-\epsilon,x+\epsilon)\right)
&=\lim_{\epsilon\downarrow 0}\lim_{n\to\infty} \frac1n \log \left(P\left(\frac1n \log U^{n} < x+\epsilon\right) 
	-P\left(\frac1n \log U^{n} \leq x-\epsilon\right)\right)\\
&=\lim_{\epsilon\downarrow 0} \lim_{n\to\infty} \frac{1}{n} \log \left( \left(1- \left(1-\frac{\lceil|\A|^{(x+\epsilon)n}\rceil}{|\A|^n}\right)^{M_n}\right) \right. \\
&\qquad\qquad\qquad\qquad \left. -\left(1- \left(1-\frac{\lceil|\A|^{(x-\epsilon)n}\rceil}{|\A|^n}\right)^{M_n}\right) \right)\\
&=\lim_{\epsilon\downarrow 0}\lim_{n\to\infty} \frac1n \log \left(|\A|^{n(\min(R+x+\epsilon-1,0))}-|\A|^{n(R+x-\epsilon-1)}\right)\\
&=-(1-R-x),
\end{align*}
as $R+x-\epsilon-1<0$ for $x\in[0,1-R]$.

The scaling result for the mean of $U^n$ follows from the application
of Varadhan's Theorem \cite{Dembo98}[Theorem 4.3.1], giving
\begin{align*}
\lim_{n\to\infty} \frac 1n \log \E(U^n)
	 = \sup_{x\in[0,1-R]}\left(x-\IU(x)\right) = 1-R.
\end{align*}
\end{proof}

\begin{figure}
\begin{center}
\includegraphics[width=0.4\textwidth]{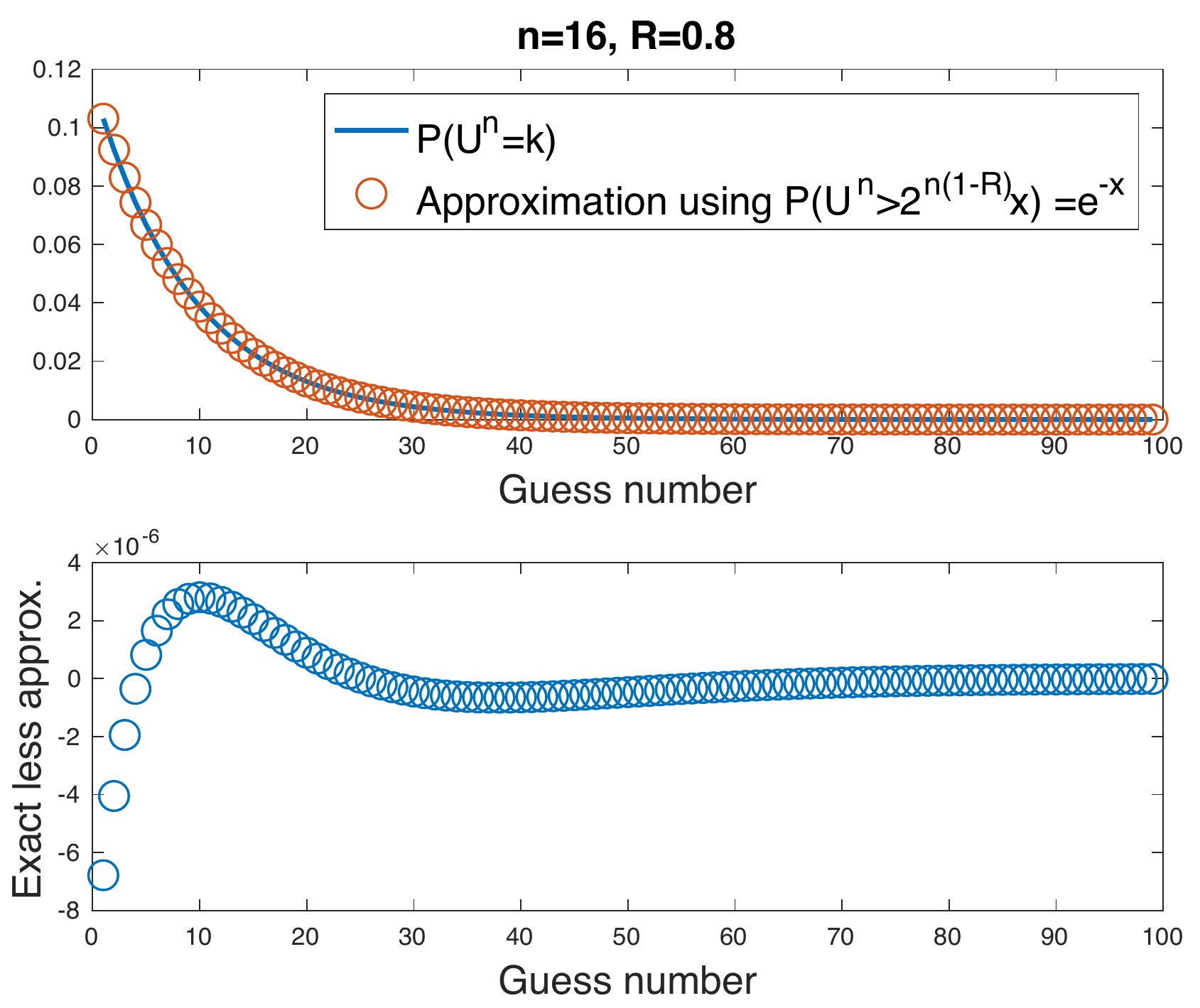}
\end{center}
\caption{Example: $\A=\{0,1\}$, block length $n=16$ and 
$R=4/5$. Upper panel: compares exact computation of
$P(U^n=k)$ (blue line) with the exponential distribution
approximation given in equation \eqref{eq:finerU} (orange
circles) for first $100$ guesses. Lower panel: the difference
between the exact and approximate values.
}
\label{fig:6_supp}
\end{figure}

Equation \eqref{eq:finerU} provides a highly accurate approximation of
the distribution of $U^n$, that it is essentially exponentially
distributed with rate $|\A|^{-n(1-R)}$ giving rise to a mean of
$|\A|^{n(1-R)}$. This is illustrated in Fig. \ref{fig:6_supp} for
a block length of $n=16$ and a  code-book of rate $R=4/5$, and becomes
more precise as $n$ increases. We will use this approximation to
make near exact computations of block error probabilities and
complexity for the BSC in Section \ref{subsec:finer}. To establish
the general channel coding and complexity results, however, it is
the LDP that is needed.  On the scale of large deviations, Theorem
\ref{theorem:LDPU} effectively says that, for large $n$, the first
non-transmitted code-word will be encountered in no more than order
$|\A|^{n(1-R)}$ guesses.

Combining Theorems \ref{theorem:LDPN} and \ref{theorem:LDPU} enables
us to provide a guessing based proof of Channel Coding Theorem.
Recalling that logarithms are taken base $|\A|$, let $h$ denote the
Shannon entropy of a random variable and let $I$ denote mutual
information. For channels introduced in equations \eqref{eq:channel}
and \eqref{eq:channel_invert}, capacity is upper bounded by $1-\HN$
as follows:
\begin{align*}
C\leq \limsup_{n\to\infty} \frac 1n \sup I(X^n;Y^n)
  \leq 1- \lim_{n\to\infty} \frac{h(\Noise^n)}{n}
  =1-\HN,
\end{align*}
where we have upper-bounded the entropy rate of the input, $h(X^n)$,
by its maximum, $n$, and used the fact that the channel is invertible
(i.e. equation \eqref{eq:channel_invert}), while the entropy rate
of the noise exists owing to to Assumption \ref{ass:N}. The proposition
that follows establishes, through the use
of a uniform-at-random code-book and GRAND, that this upper bound is achieved for all noise processes satisfying Assumption \ref{ass:N} .
We define the success rate 
\begin{align*}
\success(R) &= -\lim_{n\to\infty} \frac 1n \log P(U^n\geq\GN(\Noise^n)),
\end{align*}
which is the decay rate in the probability of correct decoding, and  evaluate it in the case where the code rate exceeds capacity.

\begin{figure}
\begin{center}
\includegraphics[width=0.49\textwidth]{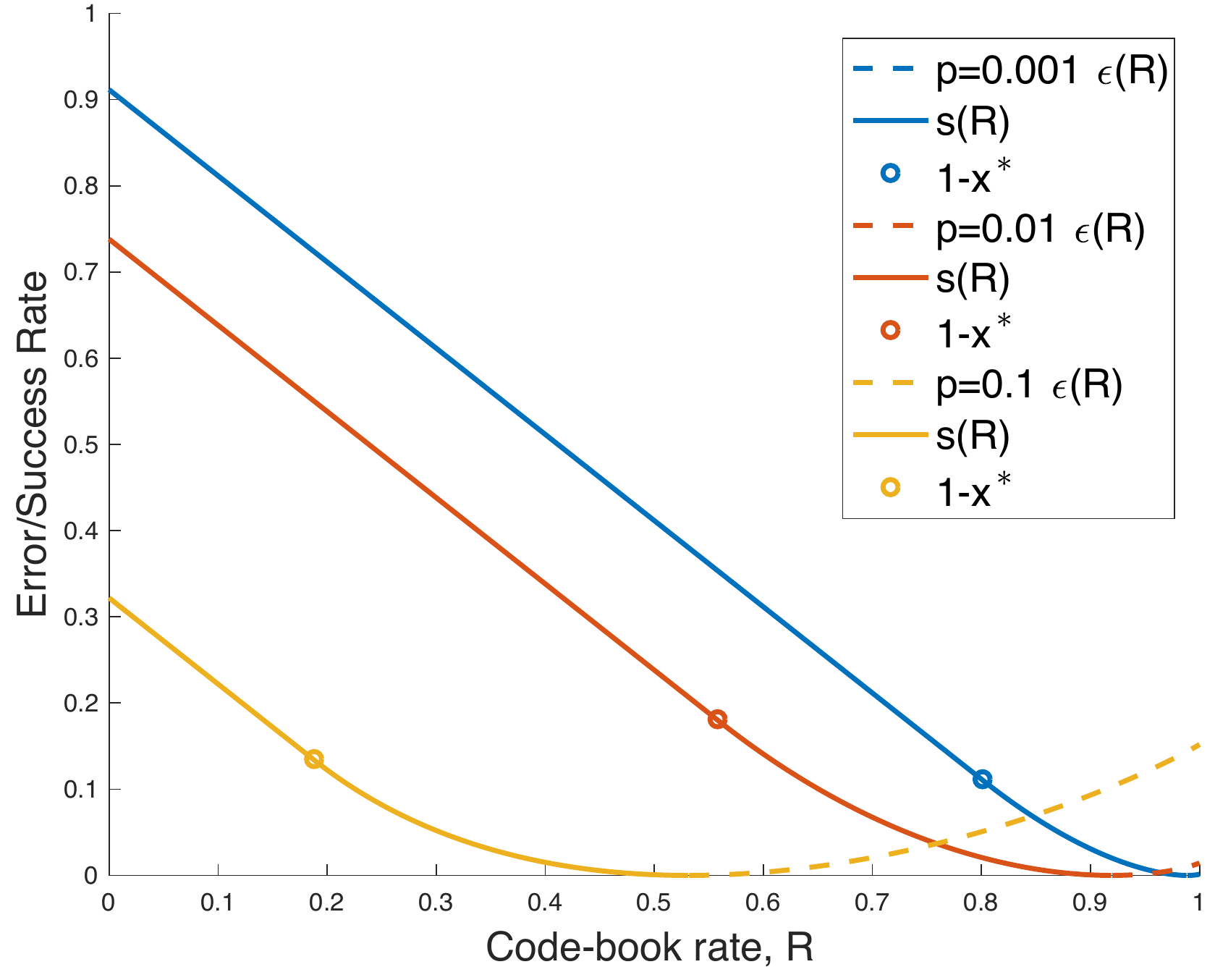}
\end{center}
\caption{GRAND decoding error and success exponents. Example: $\A=\{0,1\}$,
BSC channel, noise $\Noise^n$ made of i.i.d. Bernoulli symbols
with $P(\Noise^1=1)=p\in(0,1)$.  Code-book consisting of $M_n\approx
2^{nR}$ code-words, uniformly selected in $\A^n$. When the code-book
rate, $R$, is less than channel capacity, $1-\HN$, the probability
that a code-word that was not sent is encountered during noise
guessing before the transmitted code-word, $P(U^n<\GN(\Noise^n))$,
decays exponentially in block length $n$ with rate $\error(R)$ given
by the solid line as determined by equation \eqref{eq:error}, which
coincides in this case with Gallager's error exponent. The
point $1-\xstar$ marks the critical rate where the error-rate changes
from linear to strictly convex. For code-books rates that are beyond
capacity, $R>1-\HN$, the probability that the transmitted code-word
is identified before a non-transmitted code-word, $P(\GN(\Noise^n)<U^n)$,
decays exponentially in $n$ with rate $\success(R)$ from equation
\eqref{eq:success}, indicated by the dashed line.  
} 
\label{fig:2}
\end{figure}

\begin{proposition}[Channel Coding Theorem with GRAND]
\label{prop:channel_coding_theorem}
Under Assumptions \ref{ass:N} and \ref{ass:U}, with $\IU$ defined in equation \eqref{eq:IU}
and $\IN$ in equation \eqref{eq:IK}, 
we have the following.

1) If the code-book rate is less than the capacity,
$R<1-\HN$, then 
\begin{align*}
    \lim_{n\to\infty} \frac 1n \log P(U^n\leq\GN(\Noise^n))=
-\inf_{a\in[\HN,1-R]}\{\IU(a)+\IN(a)\}<0,
\end{align*}
so that the
probability that GRAND does not correctly identify the transmitted code-word
decays exponentially in the block length $n$.
If, in addition, $\xstar$ exists such that 
\begin{align}
\label{eq:xstar}
\frac{d}{dx}\IN(x)|_{x=\xstar}=1,
\end{align}
then
%recalling $\Hhalf$ is the R{\'e}nyi entropy rate of the noise
%with parameter $1/2$, 
the error rate simplifies further to
\begin{align}
\label{eq:error}
    \error(R) &= -\lim_{n\to\infty} \frac 1n \log P(U^n\leq\GN(\Noise^n)) 
	= \begin{cases}
	1-R-\Hhalf & \text{ if } R\in(0,1-\xstar)\\
	\IN(1-R) & \text{ if } R\in[1-\xstar,1-\HN).\\
	\end{cases}
\end{align}
Moreover, 
\begin{align*}
    \success(R) = \lim_{n\to\infty} \frac 1n \log P(U^n\geq\GN(\Noise^n))=0
\end{align*}
so that the probability
that GRAND does not provide the true channel does not decay exponentially in $n$.

2) If, instead, the code-book rate is greater than the
capacity, $R>1-\HN$, then the probability of an error is not
decaying exponentially in $n$,
\begin{align*}
    \lim_{n\to\infty} \frac 1n \log P(U^n\leq\GN(\Noise^n))=0.
\end{align*}
However,
\begin{align}
\label{eq:success}
\success(R) = \IN(1-R),
\end{align}
is strictly positive, so that the
probability that decoding produced by GRAND is the transmitted code-word does decay exponentially in $n$.
\end{proposition}

\begin{proof}
As $\{1/n \log\GN(\Noise^n)\}$ and $\{1/n \log U^n\}$ are independent
processes, $\{(1/n \log\GN(\Noise^n), 1/n \log U^n)\}$ satisfies the LDP
with rate function $\IN(x)+\IU(y)$. The LDP for $\{1/n \log
U^{n}/\GN(\Noise^n)\}$ then follows from an application of contraction principle,
\cite{Dembo98}[Theorem 4.2.1],
with the continuous function $f(x,y)=x-y$, giving
\begin{align*}
    \IUN(x) &= \inf_{a,b}\left\{\IN(a)+\IU(b): f(a,b)=a-b=x\right\}
		 = \inf_{a\in[0,1-R]}\{\IU(a)+\IN(a-x)\}.
\end{align*}

Noting the following equality
\begin{align*}
    P(U^n\leq \GN(\Noise^n)) &= P\left(\frac1n \log \frac{U^n}{\GN(\Noise^n)} \leq0\right),
\end{align*}
we can use the LDP for $\{1/n \log U^n/\GN(\Noise^n)\}$ to
determine asymptotics for the likelihood that fewer queries are
necessary to determine a non-transmitted element of the code-book
than the truly transmitted element. From the LDP lower and upper bounds,
\begin{align*}
-\inf_{x< 0} \IUN(x)
	&\leq\liminf_{n\to\infty} \frac 1n \log P\left(\frac 1n \log\frac{U^n}{\GN(\Noise^n)}\leq 0\right)
	\leq\limsup_{n\to\infty} \frac 1n \log P\left(\frac 1n \log\frac{U^n}{\GN(\Noise^n)}\leq 0\right)
	\leq -\inf_{x\leq 0} \IUN(x).
\end{align*}
For the limit to exist, we require that $\inf_{x< 0}\IUN(x)=\inf_{x\leq
0}\IUN(x)$.  Consider $\IUN(0) = \inf_{a\in[0,1-R]}\{\IU(a)+\IN(a)\}
= \IU(a^*)+\IN(a^*)<\infty$, where $a^*$ necessarily exists as $\IU$
and $\IN$ are lower-semicontinuous. As we have assumed $\HN>0$,
$a^*>0$ and $\IU(a^*)+\IN(a^*)$ is then arbitrarily well approximated
by $\IU(a^*)+\IN(a^*-\epsilon)$ as $\IN$ is continuous where it is
finite, so the above limit exits. The following simplification
is achieved by changing the order the infima are taken in:
\begin{align}
    \lim_{n\to\infty} \frac 1n \log P(U^n\leq\GN(\Noise^n)) 
	= -\inf_{x\leq 0}\IUN(x) 
	= -\inf_{x\leq 0}\inf_{a\in[0,1-R]} \{\IU(a)+\IN(a-x)\}
	= -\inf_{a\in[0,1-R]} \{\IU(a)+\inf_{y\geq a}\IN(y)\}.
	\label{eq:channel_coding_theorem1}
\end{align}

Starting from
\begin{align*}
    P(U^n\geq \GN(\Noise^n)) &= P\left(\frac1n \log \frac{U^n}{\GN(\Noise^n)} \geq0\right),
\end{align*}
similar logic, but with an additional simplification due to the form
of $\IU$ found equation \eqref{eq:IU}, leads us to 
\begin{align}
    \lim_{n\to\infty} \frac 1n \log P(U^n\geq\GN(\Noise^n)) 
	&=-\inf_{x\geq 0}\IUN(x) 
	= -\inf_{x\geq 0}\inf_{a\in[0,1-R]} \{\IU(a)+\IN(a-x)\}\nonumber\\
	&= -\inf_{a\in[0,1-R]} \{\IU(a)+\inf_{y\leq a}\IN(y)\}
	= -\inf_{x\in[0,1-R]}\IN(x). \label{eq:channel_coding_theorem2}
\end{align}

(a) For the within-capacity result, if $R<1-\HN$, then $\HN<1-R$.
Considering the right hand side of equation
\eqref{eq:channel_coding_theorem1} as both $\IU$ and $\IN$ are
decreasing on $[0,\HN]$ and $\IN$ is either infinite or increasing
on $[\HN,1-R]$,
\begin{align*}
\inf_{a\in[0,1-R]}\{\IU(a)+\inf_{y\geq a} \IN(y)\}
= \inf_{a\in[\HN,1-R]}\{\IU(a)+\IN(a)\}.
\end{align*}
This quantity is strictly positive, as $\IU$ is strictly decreasing
to zero on $[\HN,1-R]$, while $\IN$ is strictly increasing
from zero on the same range. To get the additional simplification
to equation \eqref{eq:error}, note that, as $\IN$ is strictly
convex to the right of $\HN$, $\IU$ is decreasing at rate $1$
and $\xstar$ is defined to be the value at which $\IN$ is increasing
with rate $1$, then $\inf_{a\in[\HN,1-R]}\{\IU(a)+\IN(a)\}$
is either $\IN(1-R)$ if $\xstar>1-R$ or $\IU(\xstar)+\IN(\xstar)$.
Now $\IN(\xstar) = \xstar-\Hhalf$, so that $\IU(\xstar)+\IN(\xstar)
=1-R-\xstar+\xstar-\Hhalf$ and the result follows. On the other hand,
\begin{align*}
	\inf_{x\in[0,1-R]}\IN(x) = \IN(\HN)=0
\end{align*}
and so the right hand side of equation \eqref{eq:channel_coding_theorem2}
is zero. 

(b) For the beyond-capacity result if, alternatively, $R>1-\HN$, then
$\HN>1-R$ and
\begin{align*}
\inf_{a\in[0,1-R]}\{\IU(a)+\inf_{y\geq a} \IN(y)\}
	= \IU(1-R)+\IN(\HN)=0,
\end{align*}
and so the right hand side of equation \eqref{eq:channel_coding_theorem1}
is zero. While 
\begin{align*}
	\inf_{x\in[0,1-R]}\IN(x) = \IN(1-R)>0,
\end{align*}
so that the right hand side of \eqref{eq:channel_coding_theorem2}
is strictly greater than zero. 
\end{proof}

Proposition \ref{prop:channel_coding_theorem} not only proves the
Channel Coding Theorem, but also provides exact asymptotic error
exponents when the rate of the code-book, $R$, is within capacity,
$1-\HN$, and success exponents for when the rate is beyond capacity.
For memoryless channels, the error rate in equation \eqref{eq:error}
coincides with that in \cite{Gallager65}[Theorem 2], where the
linear followed by strictly convex phenomenon was first identified.
Proposition \ref{prop:channel_coding_theorem} establishes that
phenomenon for more general noise processes.

The point $1-\xstar$ in equation \eqref{eq:error}, where the error
exponent goes from being linear in the code-book rate to strictly
convex in equation \eqref{eq:error}, is dubbed the critical rate by
Gallager for memoryless channels and can be given a simple
interpretation in terms of the noise guessing GRAND undertakes
for general noise processes. For code-book rates $R$ beyond the
critical rate, in the large $n$ limit an error occurs because the
uniform code-book is typical, but the noise is exceptionally unlikely
and far down the guessing order. For code-book rates below the
critical rate, it requires an average number of guesses to
identify the true noise, which is why the R{\'e}nyi entropy rate
with parameter $1/2$ appears, but the uniform code-book has an
unusually early entry in the noise-guessing ordered list, resulting
in an error.

Proposition \ref{prop:channel_coding_theorem} also provides success
exponents for rates above capacity. Here the interpretation of the
success rate in equation \eqref{eq:success} is that if the code-book
rate $R$ is too high for capacity, $1-\HN$, in the large $n$ limit
a successful decoding will occur if the non-transmitted code-book
elements are typically distributed, but the noise is unusually
highly likely, such that it is identified first, just prior to a
non-transmitted element of the code-book.

\begin{example*}
For the BSC, example plots of these curves are provided in Fig.
\ref{fig:2}. Note that as $\IN$ is a convex function that is zero
at $\HN$, the error and success exponents are both smooth, near-zero
functions around capacity, $R=1-\HN$.  This suggests that GRAND experiences graceful degradation
in performance near capacity.
\end{example*}

We can combine Theorems \ref{theorem:LDPN} and \ref{theorem:LDPU}
in a distinct way to determine the asymptotic complexity of the
new ML decoding scheme in terms of the number of guesses until an
ML decoding, correct or incorrect, is identified: 
\begin{align}
\GML^n := \min(\GN(\Noise^n),U^n). \label{eq:GML}
\end{align}
That is, GRAND terminates at either identification of the
noise that was in the channel or when a non-transmitted element of
the code-book is unintentionally identified, whichever occurs first.
On the scale of large deviations, if the code-book is within capacity,
$R<1-\HN$, then the sole impact of the code-book is to curtail
excessive guessing when unusual noise occurs.

\begin{figure*}
\begin{center}
\includegraphics[width=0.47\textwidth]{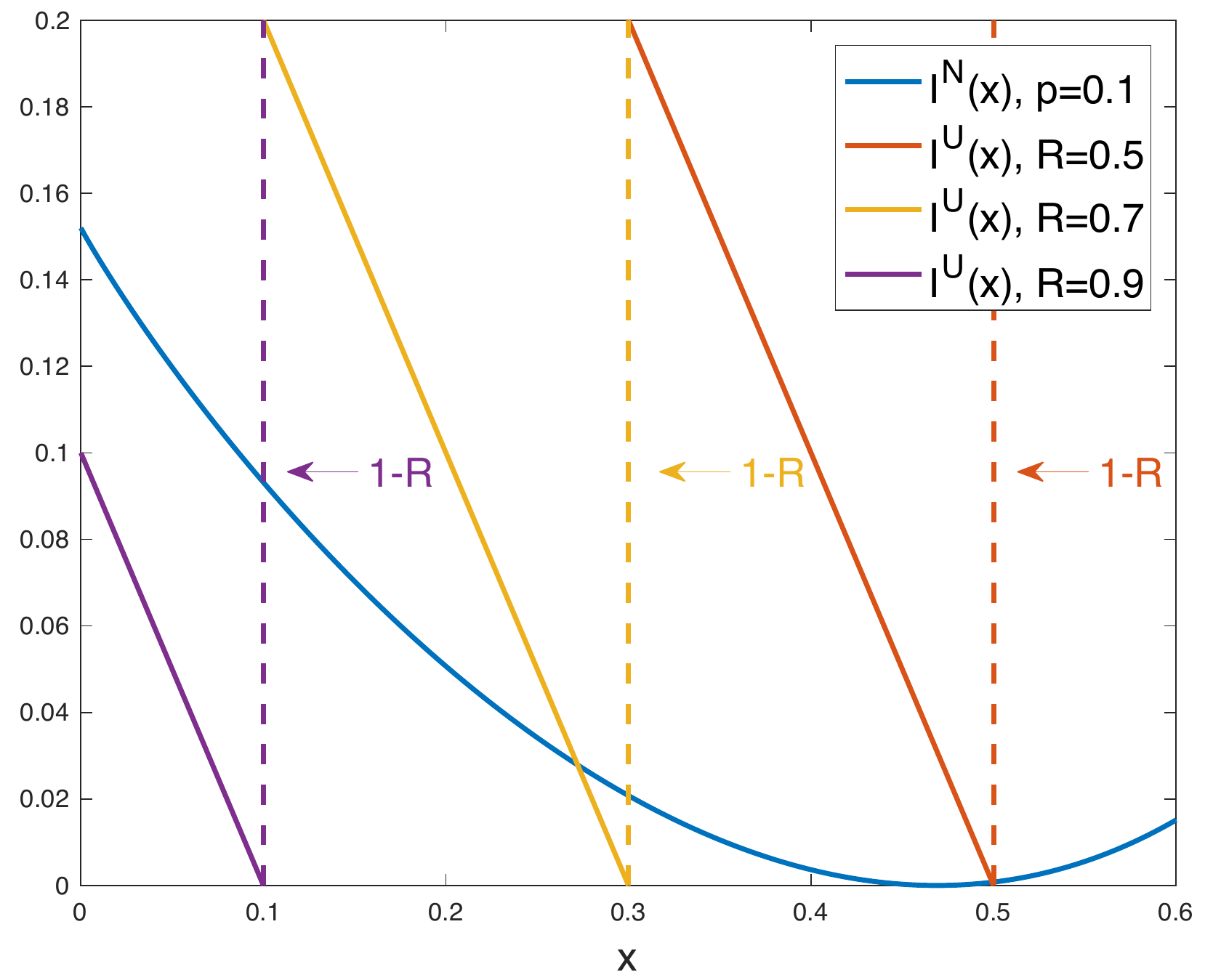}
\includegraphics[width=0.484\textwidth]{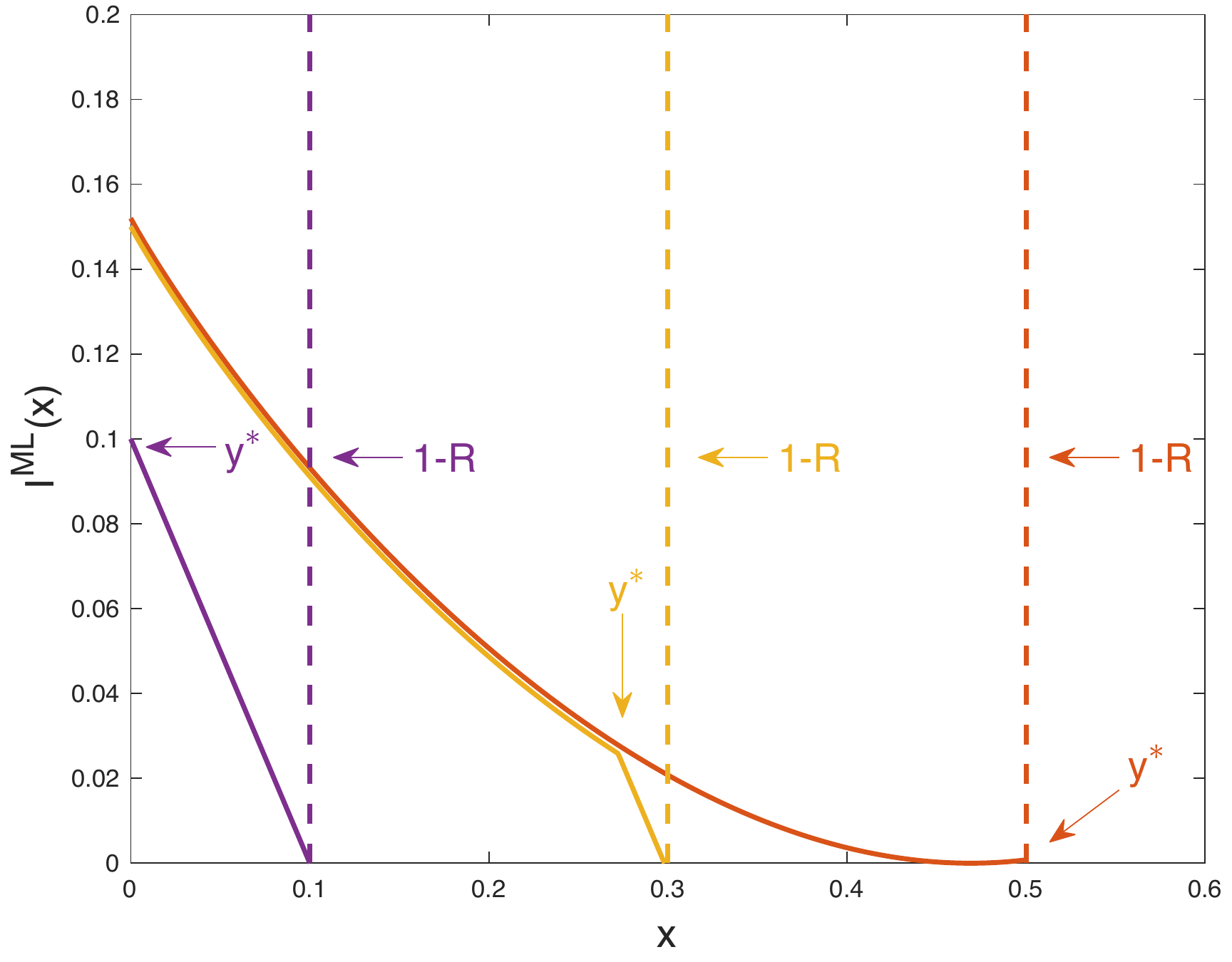}
\end{center}
\caption{GRAND complexity. Example:
$\A=\{0,1\}$, BSC channel, noise $\Noise^n$ made of i.i.d. Bernoulli
symbols with $P(\Noise^1=1)=0.1$, and channel capacity is
approximately $0.53$. Code-book consisting of $M_n\approx 2^{nR}$
code-words, uniformly selected in $\A^n$. Left panel: rate function,
$\IN$ defined in equation \eqref{eq:IK}, for the number of guesses
until the noise is identified $\{1/n\log\GN(\Noise^n)\}$. Also
plotted is the rate function $\IU$ defined in equation \eqref{eq:IU}
for the number of guesses until a non-transmitted element of the
code-book is identified, $\{1/n\log U^n\}$. Vertical dashed lines
indicate that $\IU(x)=+\infty$ for $x$ to the right of that line.
Right panel: as established in Proposition \ref{prop:IMLE}, the
rate function, $\IMLE$, that results for the number of queries until
an ML decoding is proposed in each of those three cases.  Vertical
dashed lines indicate that $\IMLE(x)=+\infty$ for $x$ to the right
of that line. If $R<1-\HN$ (red line) so that the code-book rate
is within capacity, the zero of $\IN$ occurs before the zero of
$\IU$ and the ML decoding mimics the number of guesses until the
transmitted word is identified, but with the rate function curtailed
at $1-R$. With $1-\HN<R<1-\Hmin$ (yellow line), if the algorithm
completes before $x^*$ such that $\IN(x^*)=\IU(x^*)$, whose likelihood
is decaying exponentially in $n$, the true code-word dominates, but
ultimately a non-sent code-word is returned. If $R>1-\Hmin$ (purple
line), then in this limit, an erroneous code-word is always returned.
The super-critical guessing point $\ystar$, which is the supremum over all
$y$ satisfying the conditions of Theorem \ref{theorem:WLLN}, marks
the greatest threshold below which, should the ML algorithm declare
a decoding has been found, in the large block-length limit, it will
be correct, even if the code-book rate is greater than capacity.
}
\label{fig:3}
\end{figure*}

\begin{proposition}[Guessing Complexity of GRAND]
\label{prop:IMLE}
Under Assumptions \ref{ass:N} and \ref{ass:U}, $\{1/n \log
\GML^n\}$ satisfies a LDP with a lower-semicontinuous rate function,
$\IMLE$.

1) If $R<1-\HN$, then the input code-word
will be recovered in the large deviations limit with unaffected
likelihoods, and the impact of the code-book is to curtail guessing
of unlikely inputs:
\begin{align}
    \IMLE(x) = 
	\begin{cases}
    \IN(x) & \text{ if } x\in[0,1-R]\\
    +\infty & \text{ if } x>1-R.
    \end{cases}
\label{eq:IMLEstable}
\end{align}
The average number of guesses until GRAND finds a decoding satisfies
\begin{align*}
	\lim_{n\to\infty}\frac 1n \log E(\GML^n) = \min\left(\Hhalf,1-R\right).
\end{align*}

2) If $R>1-\HN$, the code-book rate is higher than capacity and
\begin{align}
    \IMLE(x) =       
	\begin{cases}
    \min\left(\IN(x),\IU(x)\right) & \text{ if } x\in[0,1-R]\\
    +\infty & \text{ if } x>1-R.
    \end{cases}
	\label{eq:IMLEunstable}
\end{align}
This rate function need not be convex, and whichever of $\IN$ 
or $\IU$ is smaller dictates whether the ML decoding is the true
code-word or a non-transmitted one. The average number of guesses
until GRAND identifies a decoding  is governed by the beyond-capacity 
code-book rate,
\begin{align*}
        \lim_{n\to\infty}\frac 1n \log E(\GML^n) = 1-R.
\end{align*}
\end{proposition}
\begin{proof}
As $\{1/n \log\GN(\Noise^n)\}$ and $\{1/n \log U^n\}$ are independent
processes, $\{(1/n \log\GN(\Noise^n), 1/n \log U^n)\}$ satisfies the LDP
with rate function $\IN(x)+\IU(y)$. The LDP for 
$\{1/n \log \GML^n = 1/n \log\min(\GN(\Noise^n), U^n)\}$
follows from an application of contraction principle,
\cite{Dembo98}[Theorem 4.2.1],
with the continuous function $f(x,y)=\min(x,y)$, giving
\begin{align}
    \IMLE(x) &= \inf_{a,b}\left\{\IN(a)+\IU(b): f(a,b)=\min(a,b)=x\right\} 
		= \min\left\{\IN(x)+\inf_{y\geq x}\IU(y), \inf_{y\geq x} \IN(y)+\IU(x)\right\} \nonumber\\
		& = \min\left\{\IN(x), \inf_{y\geq x} \IN(y)+\IU(x)\right\}, \label{eq:IMLE}
\end{align}
where the last line follows from the form of $\IU$ in equation \eqref{eq:IU}.

The simplification of equation \eqref{eq:IMLE} into \eqref{eq:IMLEstable}
and \eqref{eq:IMLEunstable} come about about owing to considerations
from the following structure. By Theorem \ref{theorem:LDPN}, the
noise guessing rate function starts at the min-entropy rate, $\IN(0)
= \Hmin$. As the min-entropy rate is always less than or equal to
the Shannon rate, $\Hmin\leq\HN$, $\IN(\HN)=0$ and $\IN$ is convex,
$\IN$ cannot lie above line from $(0,\Hmin)$ to $(\HN,0)$.

If $R<1-\HN$, then $\HN<1-R$ and
$\IN(x)\leq\IU(x)$ for all $x\leq \HN$ from the definition of
$\IU$ in equation \eqref{eq:IU}.  For $\HN\leq x \leq 1-R$,
$\IN$ is non-decreasing and so $\min(\IN(x), \inf_{y\geq x}
\IN(y)+\IU(x))=\IN(x)$.

If, instead, $R>1-\HN$, then $1-R<\HN$ and $\inf_{y\geq x}
\IN(y)=0$ for all $x\leq 1-R$, so that $\IMLE(x)
=\min\left\{\IN(x),\IU(x)\right\}$.

To obtain the scaling result for $E(\GML^n)$ we reverse the
transformation from the rate function $\IMLE$ to its Legendre-Fenchel
transform, the scaled cumulant generating function of the process
$\{n^{-1}\log \GML^n\}$ via Varadhan's Theorem \cite{Dembo98}[Theorem
4.3.1]. In particular, note that, regardless of whether $\IMLE$ is
convex or not,
\begin{align*}
        \lim_{n\to\infty}\frac 1n \log E(\GML^n) &=
        \lim_{n\to\infty}\frac 1n \log E\left(|\A|^{\log\GML^n}\right) 
	= \sup_{x\in\R} \{x-\IMLE(x)\}.
\end{align*}
If $R<1-\HN$, this equals $\min(\Hhalf,1-R)$, while if
$R>1-\HN$ it equals $1-R$.
\end{proof}

One interpretation of the first part of that proposition is that
if the code-book is such that $R<1-\HN$, and so within capacity,
identification of the correct code-word occurs because it is likely
that all elements in the typical set of the noise will be queried
before a non-transmitted element of the code-book is identified.
Owing to the long tail of guesswork, in the absence of the other
elements of the code-book stopping the guessing algorithm, the
average number of guesses that would be made would grow with rate
$\Hhalf$ \cite{Arikan96}. If, however, one minus the normalized code-book
rate $R$ is less than that, the long tail of the scheme is
clipped. While this clipping is not enough to make an error likely,
it is enough to reduce the average number of queries that will be
made before an element of the code-book is identified. 

\begin{example*}
An example of the range of behaviors described in Proposition
\ref{prop:IMLE} for a BSC can be found in Fig. \ref{fig:3}. The
non-convex rate function corresponds to a code-book rate beyond
capacity, $R> 1-\HN$.
\end{example*}

If the code-book rate is beyond capacity, $R>1-\HN$, then implicit
in the results of Proposition \ref{prop:IMLE} is that there are
circumstances where, conditioned on the unlikely event that the
algorithm terminates after a relatively small, but exponentially
growing, number of guesses, the decoded code-word GRAND identifies is certain to be the
transmitted code-word in the large block length limit. While this
property can appear under more nuanced circumstances, we provide
one condition where the resulting characterization is simple. Namely
if the code-book rate is between channel capacity and one minus the
min-entropy rate of the noise, $1-\HN<R<1-\Hmin$, then one can
determine an exponent below which, in the limit as the block length
becomes large, if the ML algorithm terminates after a number of
guesses below the threshold governed by that exponent, the decoded code-word
will correctly correspond to the transmitted code-word.

\begin{theorem}
\label{theorem:WLLN}
Under Assumptions \ref{ass:N} and \ref{ass:U}, if $0<y<1-R$
is such that $\IN(y)<\IU(y)$, then the probability of a correct
decoding given fewer than $|\A|^{ny}$ queries are made before
the algorithm terminates converges to $1$,
\begin{align*}
\lim_{n\to\infty} P\left(\GN(\Noise^n) < U^n
		\middle| \frac1n \log D^n \leq y\right)
	=1.
\end{align*}
Such a $y$ necessarily exists if the code-book rate is less than
one minus the min-entropy rate of the noise, $R<1-\Hmin$. 
\end{theorem}
\begin{proof}
To see that such a $y$ exists if $R<1-\Hmin$, observe that as
$R<1-\Hmin$ we have that the noise guessing rate function starts
strictly below the non-transmitted guessing rate function, $\IN(0)
= \Hmin < 1-R = \IU(0)$. As both $\IN$ and $\IU$ are continuous,
the existence of such a $y$ is guaranteed.

Defining the continuous function $f:[0,1]^2\to[0,1]^3$ by
$f(x,y)=(x,y,\min(x,y))$ 
by the contraction principle, 
\begin{align*}
\left\{\left(\frac1n \log\GN(\Noise^n), \frac1n \log U^n, \frac1n \log\GML^n\right)\right\} 
\end{align*}
satisfies the LDP with rate function
\begin{align*}
\INUD(x,y,z) = \begin{cases}
	\IN(x)+\IU(y) & \text{ if } z=\min(x,y)\\
	+\infty & \text{ otherwise }.
	\end{cases}
\end{align*}
We apply the \cite{Lewis95B}[Theorem 3.1] to establish the concentration
of measure conditioned on the rare event that the algorithm terminates
within $|\A|^{ny}$ guesses. By that theorem we have that for any
open neighborhood $B$ of $(\min(y,\HN),1-R,\min(y,\HN))$,
\begin{align*}
\lim_{n\to\infty} 
P\left(\left(\frac{\log\GN(\Noise^n)}{n}, \frac{\log U^n}{n}, \frac{\log\GML^n}{n}\right) \in B 
	\,\middle|\, \frac{\log\GML^n}{n}\leq y\right) 
=1,
\end{align*}
from which the result follows.
\end{proof}
If the code-book rate is less than capacity, Theorem \ref{theorem:WLLN} recovers
what we already knew from Proposition \ref{prop:channel_coding_theorem}:
that we have concentration of measure onto correct decodings. Even
if the code-book rate is beyond capacity, however, it establishes
that, conditioned on the algorithm terminating early, there are
circumstances where we shall have concentration onto
correct decodings. Examples to this effect are presented in the
right hand panel of Fig. \ref{fig:3}, where the supremum over all
$y$ satisfying the condition of Theorem \ref{theorem:WLLN}, $\ystar$,
which we call the super-critical guessing threshold, is marked. For
code-book rates that are greater than capacity, i.e. the left two
lines, $\ystar<\HN$ and the ML decoding is only likely to be correct
if the GRAND algorithm terminates in a number of queries in the guesswork
order that is below approximately $|\A|^{n\ystar}$.

\subsection{Approximate ML decoding with GRANDAB}
\label{subsec:AML}

While Proposition \ref{prop:IMLE} identifies the computational
complexity of GRAND and so is directly related
to the decoding algorithm, Proposition \ref{prop:channel_coding_theorem}
provides a version of the Channel Coding Theorem for ML decoding
in general. That is, it relates to the likelihood that an ML decoding
is in error, irrespective of the algorithm used to identify the ML
decoding. Its proof via noise guessing, however, suggests an
approximate ML decoding scheme, GRANDAB,  with constrained complexity.

If the code-book rate is within capacity, $R<1-\HN$, the likelihood
of erroneous decoding is strictly decaying in $n$. Essentially this
occurs as the likelihood of identifying a transmitted noise sequence
is dominated by queries to up to, and including, the Shannon Typical
Set, a fact made clear by $\IN(\HN)=0$. The expected guessing
location to the first non-transmitted element encountered is governed
by one minus the code-book rate, $\IU(1-R)=0$. Thus when $R<1-\HN$,
$\HN<1-R$ and guessing the true input dominates over identifying a
non-transmitted code-word.

That guessing the noise has a long tail beyond $\HN$ is a consequence
of large growth in the number of sequences to be queried when
compared to the rate of acquisition of probability on querying them,
leading to the undesirable $\Hhalf$ growth rate for unconstrained
noise guessing. For dense code-books, this guessing tail is clipped
with an error at $1-R$, but - despite that error - capacity is
achieved so long as the code is within capacity $R<1-\HN$.  Further
contemplation of this fact suggests the following algorithm: perform
the GRAND, but abandon guessing after
$|\A|^{n(\HN+\delta)}$ queries, for some $\delta>0$, declaring an
error. This algorithm does not implement ML decoding, but it is
still capacity achieving.

\begin{proposition}[GRANDAB Coding Theorem and Guessing Complexity]
\label{prop:AML}
Under the assumptions of Theorems \ref{theorem:LDPN} and
\ref{theorem:LDPU}. If the code-book rate is less than the capacity,
$R<1-\HN$, then the GRANDAB error rate is
\begin{align*}
   \lim_{n\to\infty} \frac 1n \log P\left(\left\{U^n\leq\GN(\Noise^n)\right\}\cup\left\{\frac1n\log \GN(\Noise^n)\geq \HN+\delta\right\}\right)
&=-\min\left\{\inf_{a\in[\HN,1-R]}\{\IU(a)+\IN(a)\},\IN(\HN+\delta)\right\}<0,
\end{align*}
so that probability that the ML decoding is not the transmitted
code-word decays exponentially in the block length $n$.
If, in addition, $\xstar$ defined in equation \eqref{eq:xstar} exists 
then this simplifies to what we call the GRANDAB error rate
\begin{align}
\label{eq:errorAML}
    \errorAML(R) = \min\left(\error(R), \IN(\HN+\delta)\right)
\end{align}
where $\error(R)$ is the ML decoding error rate in equation \eqref{eq:error}.
The expected number of guesses until GRANDAB terminates, $\{\GAML^n\}$, satisfies
%\begin{align*}
%\GAML^n := \min\left(\GN(\Noise^n),U^n,|\A|^{n(H+\delta)}\right)
%\end{align*}
\begin{align*}
\lim_{n\to\infty} \frac1n \log E(\GAML^n) = \min\left(\Hhalf,1-R,\HN+\delta\right).
\end{align*}
For rates above capacity, $R>1-\HN$, the success probability is
identical to that for ML decoding, given in equation \eqref{eq:success}.
\end{proposition}
\begin{proof}
By the principle of the largest term, \cite{Dembo98}[Lemma 1.2.15] or \cite{Lewis95A},
\begin{align*}
   &\limsup_{n\to\infty} \frac 1n \log P\left(\left\{U^n\leq\GN(\Noise^n)\right\}\cup\left\{\frac1n\log \GN(\Noise^n)\geq \HN+\delta\right\}\right)\\
   &=\max\left(\limsup_{n\to\infty} \frac 1n \log P\left(U^n\leq\GN(\Noise^n)\right),
	\limsup_{n\to\infty} \frac 1n \log P\left(\frac1n\log \GN(\Noise^n)\geq \HN+\delta\right)\right),
\end{align*}
with a similar equation holding for $\liminf$. The behavior of the
first term is identified in Proposition \ref{prop:channel_coding_theorem}.
The behavior of the second term is established directly from the
LDP in Theorem \ref{theorem:LDPN} on noting that 
\begin{align*}
\inf_{x\geq \HN+\delta} \IN(x) = \IN(\HN+\delta), 
\end{align*}
as $\IN$ is strictly increasing beyond $\HN$. Coupled with the
continuity of $\IN$, we obtain equation \eqref{eq:errorAML}. The
expected number of guesses until the algorithm completes is determined
in an identical manner to that in Proposition \ref{prop:IMLE}. 
\end{proof}
The interpretation of this result is straight-forward: GRANDAB results
in an error if either the ML decoding is erroneous, as governed
by Proposition \ref{prop:channel_coding_theorem}, or if the algorithm
abandons guessing before an element of the code-book is identified.
Whichever of these two events is more likely dominates the error
rate. So long as the algorithm does not abandon until after querying
all elements in the typical set of the noise, it is capacity
achieving.

The earlier Theorem \ref{theorem:WLLN} also suggests an abandonment
rule when code-books are at rate beyond capacity. One could curtail
querying and declare an error after approximately $|\A|^{n\ystar}$
guesses, where $\ystar$ is maximum over all $y$ satisfying the
conditions of Theorem \ref{theorem:WLLN}. Before that point, it is
likely that the decoding is correct, while afterwards it is likely
to be incorrect.

\section{Examples}
\label{sec:examples}

As all of the results in this paper hold for channels with memory,
to illustrate the complexity, error and success probabilities of
GRAND and GRANDAB decoding we treat binary $\A=\{0,1\}$ noise sequences
$\{\Noise^n\}$ whose elements are chosen via a Markov chain with
transition matrix
\begin{align*}
\left(
        \begin{array}{cc}
        1-a & a \\
        b & 1-b
        \end{array}
\right),
\end{align*}
and assume that $a,b>0$. The initial distribution of the Markov chain
can go unspecified as it plays no role in the asymptotic results.
This model includes the BSC by setting $p=a=1-b$, but, in general,
the second eigenvalue is $1-a-b$, which characterizes the burstiness,
memory or mixing of the Markov chain.

The R\'enyi entropy rate of this noise source can be evaluated
\cite{Malone04} for $\alpha\neq1$ to be
\begin{align*}
H_\alpha =& \frac{1}{1-\alpha}\log
        \left(
        (1-a)^\alpha+(1-b)^\alpha 
	+
        \sqrt{((1-a)^\alpha-(1-b)^\alpha)^2 + 4(ab)^\alpha }
        \right)
        -\frac{1}{1-\alpha}.
\end{align*}
While with $h(a) = -a\log(a)-(1-a)\log(1-a)$ being the binary Shannon entropy,
$H_1=\HN=h(a)b/(a+b)+h(b)a/(a+b)$ is the Shannon entropy rate of the Markovian source. Thus
using equation \eqref{eq:LambdaN} we have an explicit expression
for the resulting scaled cumulant generating function, $\LambdaN$,
of the logarithm of the noise. While the rate function $\IN$ defined
in equation \eqref{eq:IK} cannot be calculated in closed form, it
is readily evaluated numerically, only requiring the solution of a
one-dimensional concave optimization.

While prefactors are not captured in that asymptotic analysis in
Propositions \ref{prop:channel_coding_theorem}, \ref{prop:IMLE} and
\ref{prop:AML}, they allow the following approximations. For GRAND and
GRANDAB decoding, our measure of complexity is the average number of
guesses per bit:
\begin{align*}
\text{ GRAND ave. no. guesses / bit} &\approx \frac{2^{n\min(1-R,\Hhalf)}}{n}\\
\text{ GRANDAB ave. no. guesses / bit} &\approx \frac{2^{n\min(1-R,\Hhalf,\IN(\HN+\delta))}}{n}.
\end{align*}

For comparison, we define the complexity of the straight computation
of the ML decoding in \eqref{eq:straight_MLE} to be
the number of conditional probabilities that must be computed
per bit before rank ordering and determining the most likely code-book
element:
\begin{align*}
\text{No. conditional prob. computations / bit} &= \frac{2^{nR}}{n}.
\end{align*}
Thus we are equating the work performed in one noise guess with one
computation of a conditional probability. As this direct scheme
results in the ML decoding as by noise guessing, it shares the same
error and success probabilities as GRANDAB.

For error and success probabilities we employ:
\begin{align*}
\text{GRAND prob. of error} & \approx 2^{-n\error(R)} 
	\text{ for } R <1-\HN,\\
\text{GRANDAB prob. of error} & \approx 2^{-n\errorAML(R)} 
	\text{ for } R <1-\HN,\\
\text{GRAND \& GRANDAB prob. of success} & \approx 2^{-n\success(R)} 
	\text{ for } R >1-\HN,
\end{align*}
where $\error$, $\errorAML$, and $\success$ are given in equations
\eqref{eq:error}, \eqref{eq:errorAML} and \eqref{eq:success}.

We use the following rule to select the parameter
$\delta$ that determines how far beyond the Shannon typical set
queries are made before abandonment in GRANDAB. With the
stationary probability of noise per bit begin $p$, for a given
block-length $n$ we identify $\delta$ such that the probability of
abandonment is no more than $\pabandon$ times the expected uncoded
block error probability; i.e we solve the following equation
numerically for $\delta(n)$:
\begin{align*}
2^{-n\IN(\HN+\delta(n))} = \pabandon\min(pn,1).  
\end{align*}
Selecting this $\delta$ sets a floor for the block-error probability
generated by abandoned guessing that is a fraction of the uncoded
block-error probability.

We set $\pabandon=10^{-3}$
if the average bit error rate in the channel is $10^{-4}$ and
$\pabandon=10^{-2}$ if it is $10^{-2}$ indicating we
are willing to tolerate block-error probabilities that are of order
at least $100$ or $1000$ times less likely than an uncoded block
error.

For complexity, as the number of computations per bit per second
is normally several orders of magnitude greater than the number of
bits received over the channel per second, we will consider a
complexity feasible if it is in the range of $10^3-10^4$ guesses per
bit. For both GRAND and GRANDAB, this is likely
to be a conservative constraint as the guessing is readily
parallelizable.

\subsection{Binary Symmetric Channel (BSC)}

\begin{figure*}
\begin{center}
\includegraphics[width=0.9\textwidth]{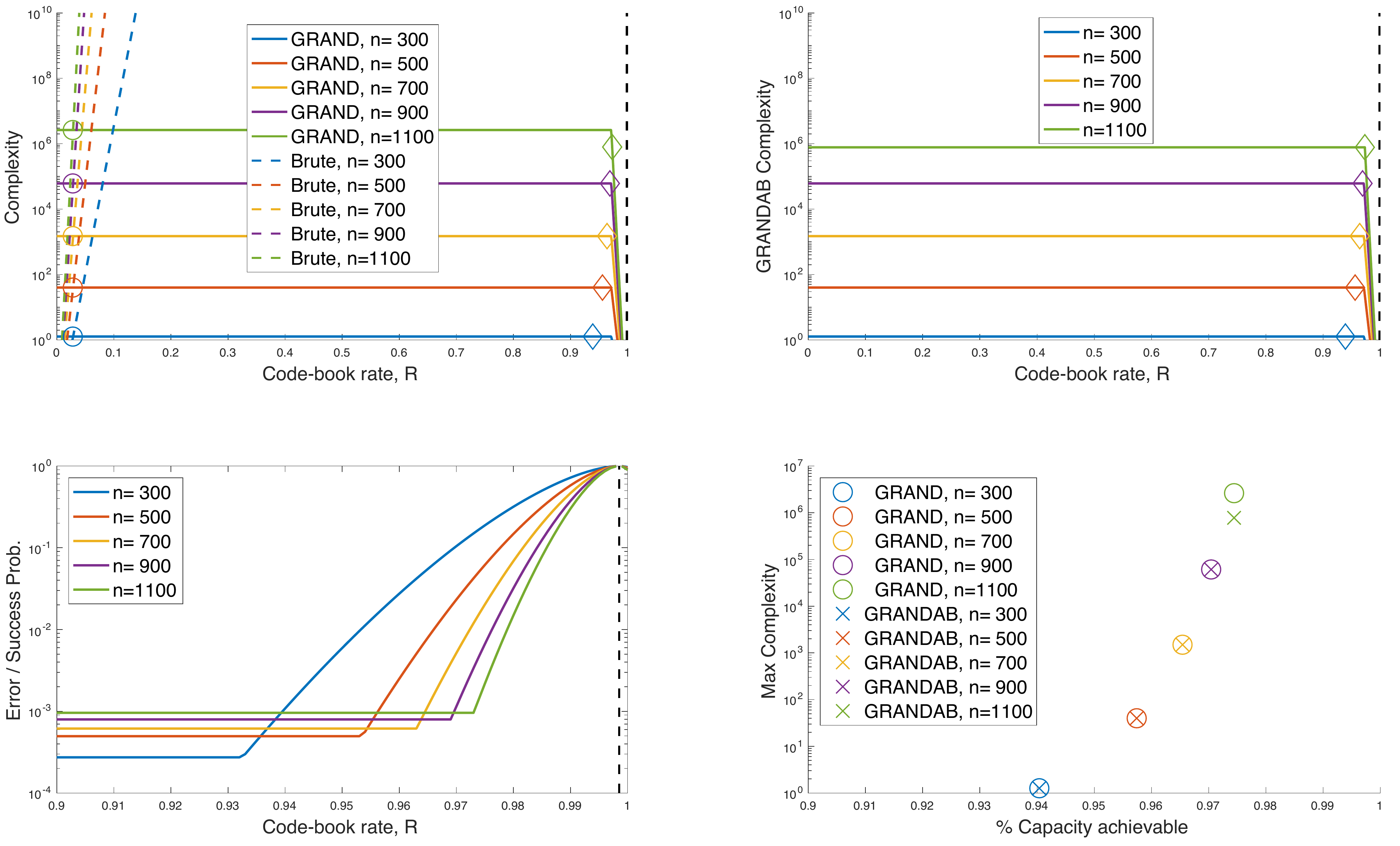}
\end{center}
\caption{
BSC GRAND and GRANDAB decoding. Bit flip probability
$p=10^{-4}$, code-book rate $R$ and block length $n$.  Dashed
vertical lines in three of the panels indicate channel capacity.
Top left panel: complexity of ML decoding by noise guessing (solid
lines) or by
brute force (dashed lines) as a function of code-book rate.  Circles indicate the
rate beyond which computing within the code-book has higher complexity
than noise guessing. Diamonds indicate the rate below which block
error probability is less than $10^{-3}$.
Top right panel: complexity of GRANDAB as a function of code-book
rate, where the free parameter $\delta$ in  GRANDAB is
selected as described in Section \ref{sec:examples}.
The inflection in complexity in these top two panels occurs at the cut-off rate.
Bottom left panel: with a zoomed in x-scale, to the left of capacity
the curves show approximate error probability of GRANDAB for a
range of $n$. To
the right of capacity the curves show approximate success probability
of both GRAND and GRANDAB.
Bottom right panel: for each block length and both GRAND and GRANDAB, the maximum achievable rate, as a percentage of 
capacity, while keeping the block error probability below $10^{-3}$
is plotted against the highest complexity of the code, which
occurs for low-rate code-books. 
}
\label{fig:4_10000}
\end{figure*}

For the BSC with bit error probability $p=10^{-4}$, a GRANDAB decoding
abandonment probability of $\pabandon=10^{-3}$, and a range of block
lengths, $n$, the approximate complexity and error performance of
GRAND and GRANDAB is shown in Fig. \ref{fig:4_10000}. 

The top left panel shows the complexity (average number of guesses
per received bit) for GRAND (solid lines)
and by brute force (dashed lines) for a range of block lengths,
$n$, with the vertical dashed line indicating capacity, $1-\HN$.
The computational complexity of the brute force approach, computing
conditional probabilities for all elements of the code-book rapidly grows with rate.
The complexity of guessing the noise only
decreases as rates increase, with the circles indicating the threshold
above which the complexity of guessing within the code-book is
less than that of brute force determination.  The diamond marks the
code-book rate after which the block error probability for GRAND is $\pblock=10^{-3}$ and so sets an upper-threshold
on the code-book rate. The top right panel shows the equivalent
complexity plot for GRANDAB decoding. The effect of abandonment is to
reduce the maximum complexity for the longest block-length, with
no impact on smaller block-lengths in this instance.

The bottom left panel shows the approximate block-error and
block-success probabilities below and above capacity, respectively,
for GRANDAB as a function of code-book rate. The ML curves would be
identical at higher rates, but would drop further at lower code-rates as the abandonment
of guessing of GRANDAB is what places a floor on the block-error rate. 

For both GRAND and GRANDAB, the final panel,
bottom right, shows the maximum complexity for a given block length,
$n$, versus the \% of capacity achievable with a code-book rate
that provides a block error probability below $\pblock=10^{-3}$.
With the rule of thumb that $10^3-10^4$ guesses per bit is acceptable,
then choosing $n=700$ could realize up to 96.5\% of capacity. Note
that this occurs for a block length that is substantially smaller
than the reciprocal of the bit error rate, $1/p = 10,000$.

The inflection in complexity for the top two panels occurs at the
cut-off rate. This illustrates an intriguing property of GRAND and
GRANDAB. While for sequential decoding of tree codes, decoding
complexity increases steeply when the rate exceeds the cut-off rate,
for decoding by guessing noise, complexity decreases past the cut-off
rate.

%\begin{figure*}
%\begin{center}
%\includegraphics[width=0.9\textwidth]{fig4_1000}
%\end{center}
%\caption{BSC ML and AML decoding. Same display
%as for Fig. \ref{fig:4_10000}, but with bit flip probability %$p=10^{-3}$
%and a block-error probability floor of $10^{-2}$. 
%}
%\label{fig:4_1000}
%\end{figure*}

\begin{figure*}
\begin{center}
\includegraphics[width=0.9\textwidth]{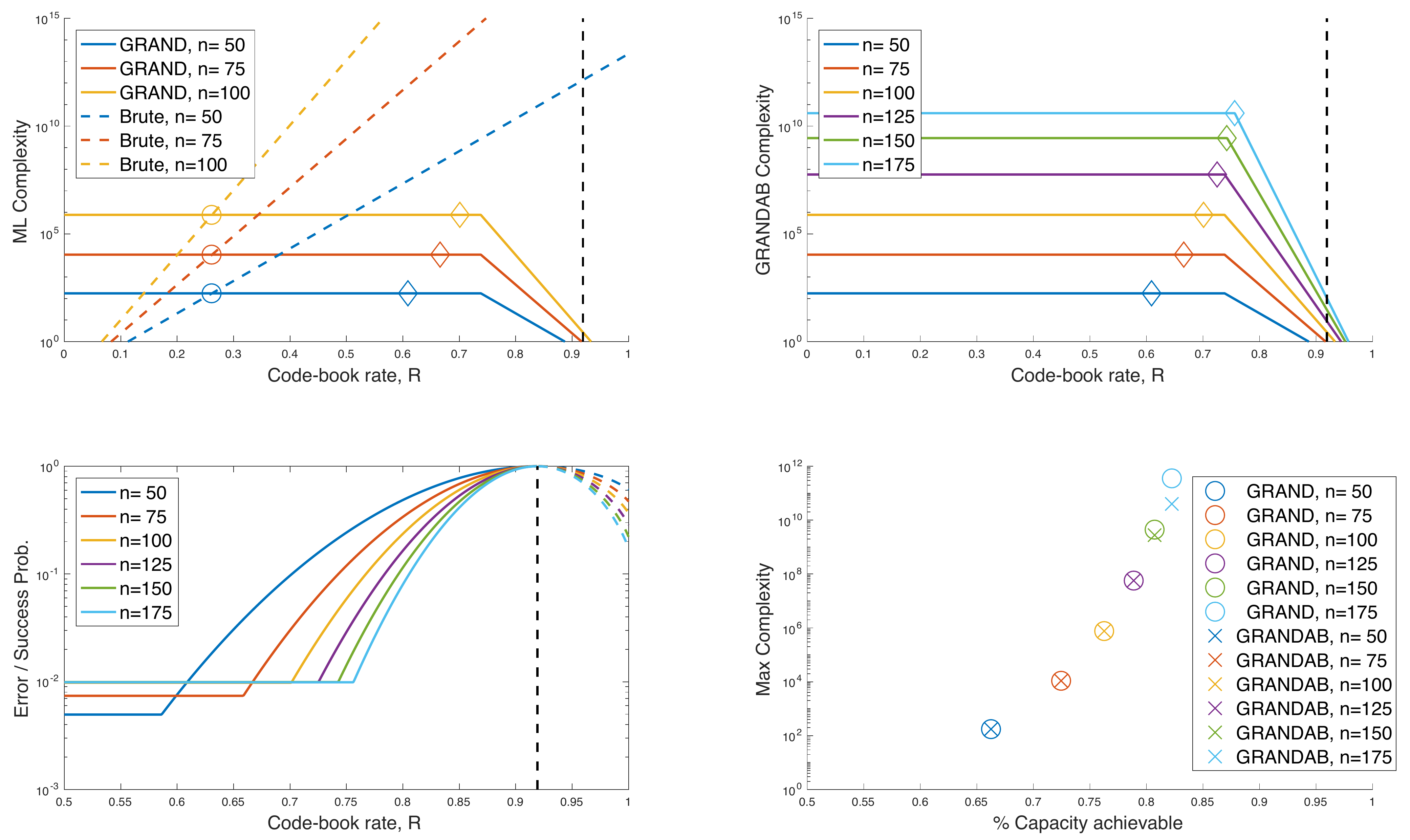}
\end{center}
\caption{BSC GRAND and GRANDAB. Same display
as for Fig. \ref{fig:4_10000}, but with bit flip probability $p=10^{-2}$
and a block-error probability floor of $10^{-2}$. 
}
\label{fig:4_100}
\end{figure*}

Analogous information is displayed for the BSC with bit error
probability 
%$p=10^{-3}$ in Fig. \ref{fig:4_1000} and 
$p=10^{-2}$
in Fig. \ref{fig:4_100}, but with $\pabandon=10^{-2}$. Again, the
computational complexity of the brute force approach makes it
infeasible even for modest rates. For these higher bit error
probabilities, the effect of  GRANDAB's truncation is felt at smaller
block sizes. This might be expected, given the Shannon entropy of
the noise has increased. As the likelihood of noise is increased,
block-lengths must be reduced to keep guesswork down to the
$10^3$--$10^4$ guesses per-bit range. 
%For $p=10^{-3}$, of block-lengths
%shown, the receiver would select $n=200$ for reasons of
%complexity. Requiring that the block error rate be below
%$\pblock=10^{-2}$, at this block length a maximum rate of 90.6\%
%of capacity possible. 
For $p=10^{-2}$, complexity considerations
reduce $n$ to $75$, for which rates providing up to 72.4\%
of capacity are achievable with a block error probability no more
than $\pblock=10^{-2}$.

\subsection{Bursty Markovian noise}

\begin{figure*}
\begin{center}
\includegraphics[width=0.9\textwidth]{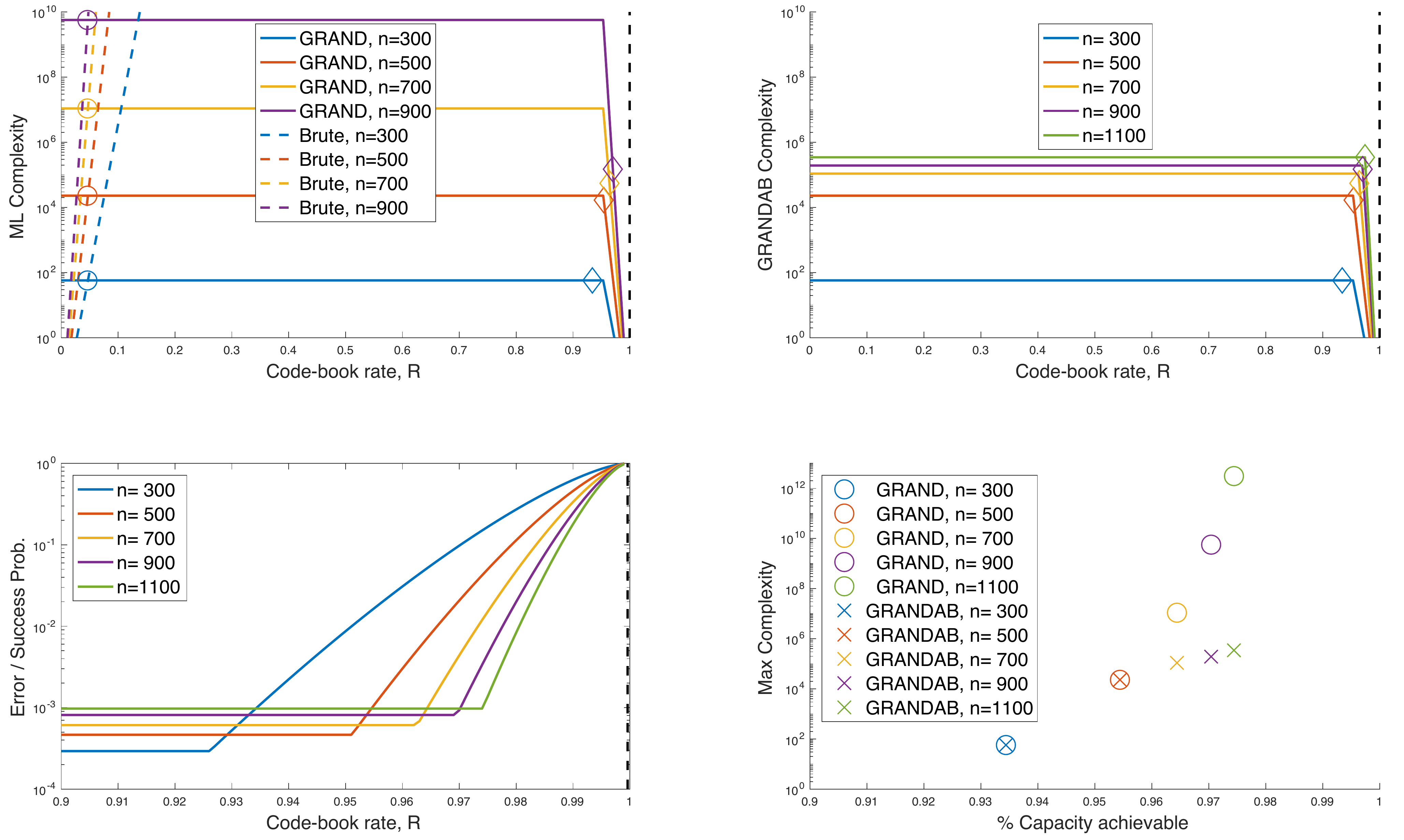}
\end{center}
\caption{GRAND and GRANDAB decoding with binary Markovian noise. Average
bit flip probability $p=10^{-4}$, making it comparable to the BSC
plots in Fig. \ref{fig:4_10000}, but for a Markovian channel with
$a=p/5=2\times 10^{-5}$ and $b = (1-p)/p a \approx  0.2$, so making an
extremely bursty noise channel.  Four displayed panels are analogous
to those described in the caption of Fig. \ref{fig:4_10000}.
}
\label{fig:5_10000}
\end{figure*}

A core feature of the proposed schemes is that they can be applied
in channels with correlated noise without the need for interleaving
and other methods that alleviate the impact of memory. The equivalent
of Fig. \ref{fig:4_10000} is presented in Fig. \ref{fig:5_10000}
where the long run average probability of bit-error is set to be
the same, $p=10^{-4}$, in both, but here $a=10^{-4}/5$ and $b=1/5$.
These have been selected to give a highly bursty source where the
likelihood of a bit flip is small, but the likelihood of an additional
bit flip given one has occurred is $3$ orders of magnitude higher.
The block-lengths displayed for the Markovian channels are the same
as for the corresponding BSC example and again $\pabandon=10^{-3}$,
to enable ready comparison.

For this parameterization, the complexity of GRAND is much higher for this Markovian noise than the BSC
equivalent. Consequently, GRANDAB plays a more significant role in
reducing that complexity for large block lengths by abandonment.
Based on the criteria set for the BSC, for reasons of complexity
$n=500$ would be selected. While this is shorter than the block
length for the equivalent BSC, it is still the case that 95.4\% of
capacity is achievable with a block error rate of less than
$\pblock=10^{-3}$.

%Similarly, Fig. \ref{fig:5_1000} is analogous to the BSC in Fig.
%\ref{fig:4_1000}, with an average bit error probability $p=10^{-3}$
%obtained by setting $a=10^{-3}/5$ and $b=a(1-p)/p\approx 0.1998$.
%Based on the complexity considerations, of those in the figure,
%$n=300$ would be selected, with AML providing a significant reduction
%of the order of $5\times 10^2$ over ML decoding by noise guessing.
%For that block-length, 92.5\% of capacity is a achievable with a
%block-error probability of no more than $\pblock=10^{-2}$, again
%comparing favorably with the BSC equivalent.

Fig. \ref{fig:5_100} can be compared with the BSC in Fig.\ref{fig:4_100},
having $p=10^{-2}$ obtained by $a=10^{-2}/5$ and $b=a(1-p)/p\approx
0.198$. For this noisy channel, again  GRANDAB provides a reduction in
algorithmic complexity at a cost of introducing an error floor. If
the receiver wishes to limit complexity, they would select $n=75$.
With a threshold of a block-error rate set at $10^{-2}$, 71.2\% of
capacity is available.

Note that in all examples presented here the best block lengths are
no larger than the reciprocal of the corresponding bit error rate,
$1/p$. This behavior may be unexpected if we consider error exponents
for Markov channels based on interleaving of the order of the mixing
time of the Markov noise model \cite{Med98}, yet it is a desirable
feature of the scheme, which we have consistently observed.

%\begin{figure*}
%\begin{center}
%\includegraphics[width=0.9\textwidth]{fig5_1000}
%\end{center}
%\caption{ML and AML decoding with Markovian noise. Same display as
%for Fig.\ref{fig:5_10000}, but with average bit flip probability
%$p=10^{-3}$.
%}
%\label{fig:5_1000}
%\end{figure*}

\begin{figure*}
\begin{center}
\includegraphics[width=0.9\textwidth]{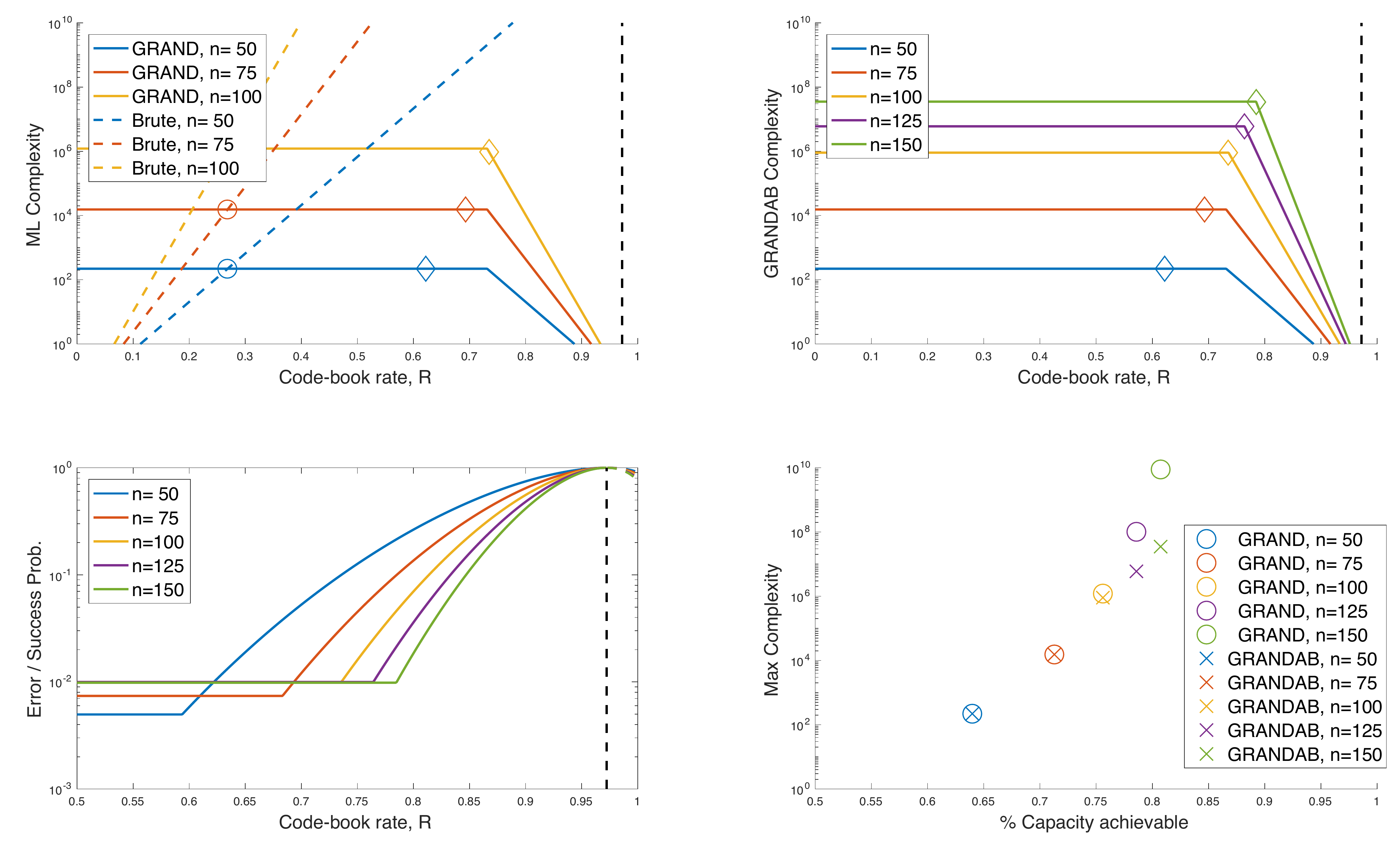}
\end{center}
\caption{GRAND and GRANDAB with Markovian noise. Same display as
for Fig.\ref{fig:5_10000}, but with average bit flip probability
$p=10^{-2}$.
}
\label{fig:5_100}
\end{figure*}

\subsection{Finer approximations for the BSC}
\label{subsec:finer}

\begin{figure*}
\begin{center}
\includegraphics[width=0.9\textwidth]{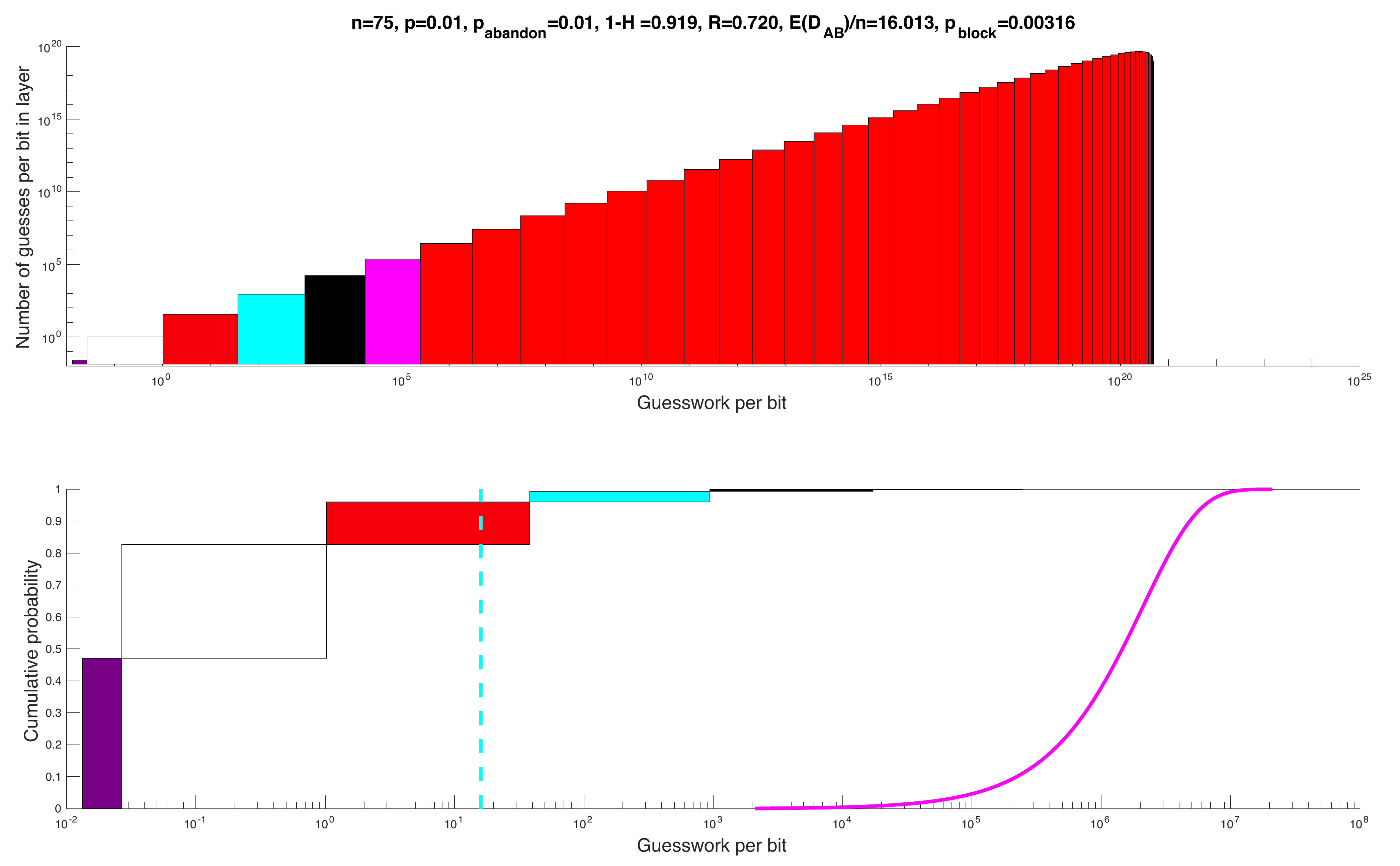}
\end{center}
\caption{
Example: $\A=\{0,1\}$, BSC channel, bit flip probability $p=10^{-2}$,
block length $n=75$, capacity $1-H=0.919$ and a code-book rate of
$R=0.72$.  Upper panel: The x-axis is the total number of queries
per-bit on a log-scale. The y-axis is the number of queries per-bit
that are made to sequences of increasing Hamming distance, also on
a log-scale. Each rectangle demarcates a distinct Hamming distance.
The color coding indicates the probability that is accumulated by
guessing through a layer of given Hamming distance and runs from
blue, $0$, to red, $1$. The white layer is where $2^{nH}$ guesses
have been made, at the core of the Shannon Typical Set. Accumulation
of probability around the white layer is asymmetric.  Prior to it,
probability is quickly obtained, but the decreasing probability per
sequence coupled with the increasing number of sequences with the
same probability results in $2^{nH_{1/2}}$, asymptotically the
average guesswork, being in the black layer. The cyan layer indicates
the guessing layer by the end of which there is a 99\% chance of
identifying true noise. Abandoning guessing here if no code-word
had been identified would result in $205$ fewer guesses per bit
than would, on average, be necessary to identify the true noise
sequence. With a code-book rate of $0.72$, using the approximation
in \eqref{eq:finerU}, the magenta layer is where a non-transmitted
code-word would, on average, be identified.  Lower Panel: Cumulative
probability of guesswork with the same colour coding as the upper
panel, but with a truncated y-axis to show more detail. The dashed
vertical cyan line is located at the average number of guesses per-bit per
GRANDAB decoding, $E(\GAML)/n\approx16$. The magenta
line is the cumulative distribution of the number of guesses per
bit until a non-transmitted code-word is identified, using the
approximation to $U^n$ found in equation \eqref{eq:finerU}.  Note
that, with a log x-scale, it is tightly centered around its mean,
resulting in a block error probability of $\pblock = 3.15 \times
10^{-3}$.  } \label{fig:6_100} \end{figure*}

%\begin{figure*}
%\begin{center}
%\includegraphics[width=0.9\textwidth]{fig6_1000}
%\end{center}
%\caption{
%Example: $\A=\{0,1\}$, BSC channel, bit flip probability $p=10^{-3}$,
%block length $n=200$, capacity $1-H=0.989$ and a code-book rate
%of $R=0.906$.
%Upper panel: Same layout as in Fig. \ref{fig:6_100}.
%Abandoning
%guessing after 99\% probability is accumulated if no code-word had %been identified would result
%in $32$ fewer guesses per bit than would, on average, be
%necessary to identify the true noise sequence. 
%Lower Panel: As in Fig. \ref{fig:6_100}. 
%Average number of guesses per-bit per AML decoding %$E(\GAML)/n\approx0.268$.
%Block error probability of $\pblock = 1.19 \times 10^{-3}$.
%}
%\label{fig:6_1000}
%\end{figure*}

\begin{figure*}
\begin{center}
\includegraphics[width=0.9\textwidth]{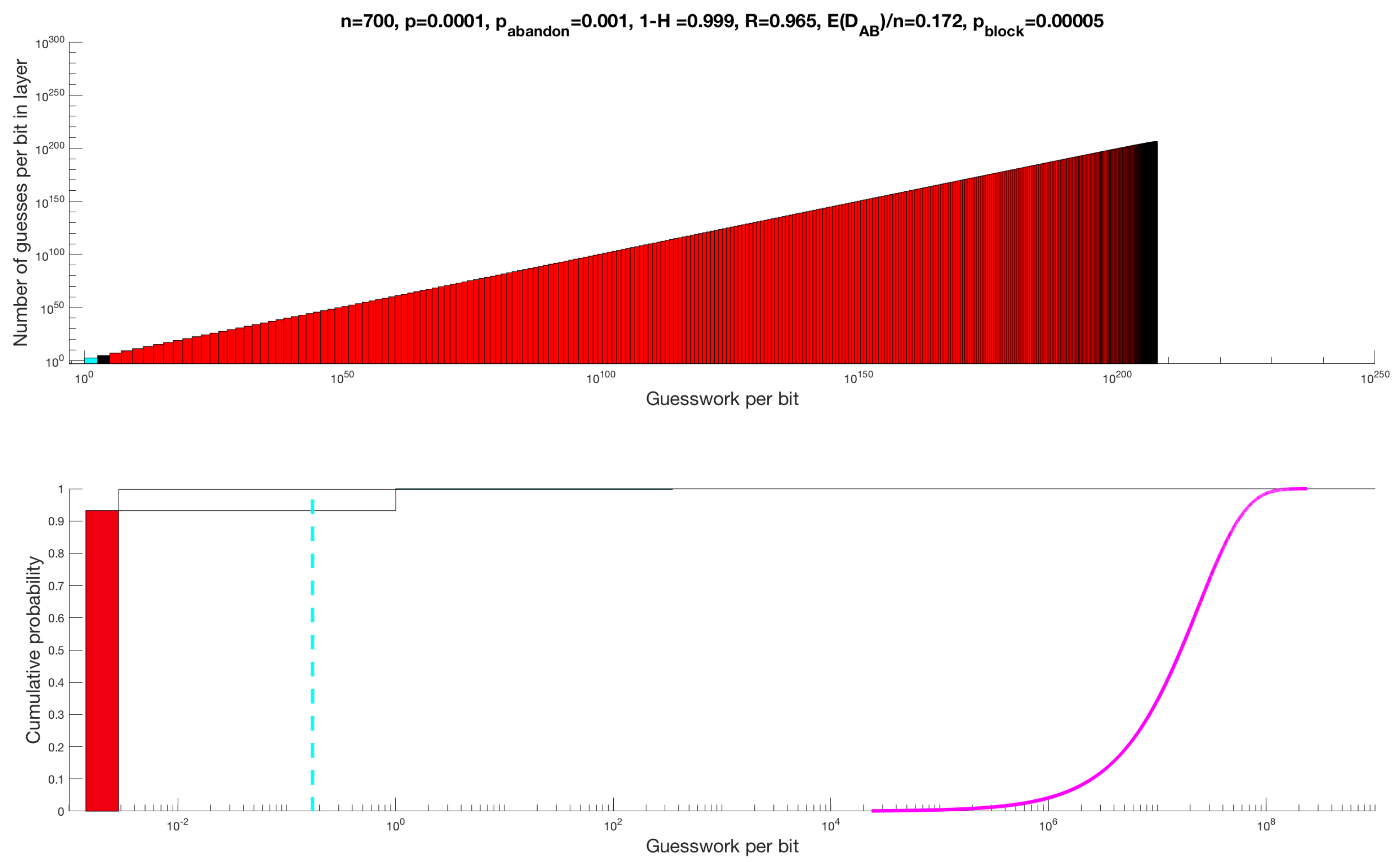}
\end{center}
\caption{
Example: $\A=\{0,1\}$, BSC channel, bit flip probability $p=10^{-4}$,
block length $n=700$, capacity $1-H=0.999$ and a code-book rate
of $R=0.965$.
Upper panel: Same layout as in Fig. \ref{fig:6_100}. Abandoning
guessing after 99.9\% probability is accumulated if no code-word
had been identified would result in $115$ fewer guesses per bit
than would, on average, be necessary to identify the true noise
sequence.  Lower Panel: As in Fig. \ref{fig:6_100}. 
Average number of guesses per-bit per GRANDAB decoding $E(\GAML)/n\approx0.172$.
Block error probability of $\pblock = 4.69 \times 10^{-5}$.
}
\label{fig:6_10000}
\end{figure*}

For uniform-at-random code-books, Proposition
\ref{prop:channel_coding_theorem} provides error exponents for
general noise processes. In the case of the memoryless channel,
however, a more exact computation of the block error probability
is possible. This is achieved by availing of the precision of the
finer approximation to the distribution of the number of guesses
until a non-transmitted code-word is identified, $U^n$, given in
equation \eqref{eq:finerU}.

The error probability is one minus the success probability,
\begin{align*}
P(U^n\leq G(\Noise^n)) = 1-P(G(\Noise^n)<U^n), 
\end{align*}
and we shall provide a more exact computation of the latter.
Restricting to a BSC, there are $n$ choose $0$ noise strings with
no errors, $n$ choose $1$ strings with one error, and so forth.
Thus we define $l_{-1}=0$ and
\begin{align*}
l_k = {n \choose 0} + {n \choose 1} + \cdots + {n \choose k}
\end{align*}
for each $k\in\{0,\ldots,n\}$.
Consequently in guesswork order we have
\begin{align*}
P(G(N^n)=m) = p^k(1-p)^{n-k} 
\end{align*}
for every $m\in\{l_{k-1}+1,\ldots,l_k\}$. Thus
\begin{align*}
P(G(\Noise^n)<U^n)= \sum_{k=0}^n p^k(1-p)^{n-k} \sum_{m=l_{k-1}+1}^{l_k} P(U^n>m).
\end{align*}
Approximating the distribution of $U^n$ by
\begin{align*}
P(U^n>m) \approx \exp(-m2^{-n(1-R)}),
\end{align*}
as suggested by equation \eqref{eq:finerU},
and computing the resulting geometric sum gives 
\begin{align}
P(G(\Noise^n)<U^n)\approx 
&\sum_{k=0}^n p^k(1-p)^{n-k} \left(\frac{e^{-(l_{k-1}+1)2^{-n(1-R)}}-e^{-(l_k+1)2^{-n(1-R)}}}{1-e^{2^{-n(1-R)}}}\right).
\label{eq:BSC_PBLER}
\end{align}
Thus for a BSC we can compute a finer approximation to the block error probability,
$\pblock$, by a sum of only $n+1$ terms.

Fig. \ref{fig:6_100} reconsiders the scenario treated via the large
deviations analysis in Fig. \ref{fig:4_100}, but with this finer
approximation for the block error probability. The $n$ and $R$ used
correspond to those deduced from the asymptotic analysis as maximizing
rate subject to constraints on block error probability while
maintaining a certain degree of complexity.  The true block error probability
is $3\times 10^{-3}$, when the target in the asymptotic regime was
$10^{-2}$ indicating good accuracy.

In all
cases we have examined beyond those shown here, the asymptotic results compare well with the more precise
computations which, if anything, suggest that higher rates can be
obtained while still meeting block error targets.

\section{Discussion and Conclusions}
\label{sec:conc}

We have introduced and analyzed two decoding algorithms based on
guessing that are suitable for a broad class of noise processes.
Subtracted noise from a received signal in order from most likely
to least likely, the first instance that is in the code-book
corresponds to the ML decoding. Both GRAND, which identifies an ML decoding by noise guessing,
and  GRANDAB, an approximate ML decoding by noise guessing algorithm in which the receiver quits its
attempts to identify an element of the code-book after a given
number of unsuccessful queries that is determined by the Shannon
entropy of the noise, are capacity achieving when used with
uniform-at-random code-books. 
To establish capacity results, we have assumed that the source is uniform. If, at the cost of capacity, the source is not uniform, the decoding algorithm remains unchanged. 
Depending on channel conditions, GRANDAB
has the potential benefit over ML decoding of decreased complexity,
even for DMCs. Analytically leveraging this noise-focused view,
we provide explicit error and success exponents for code-book rates
that are within and beyond capacity, respectively, providing a
version of the Channel Coding Theorem.

While DMCs form the classic model in information theory, real
communication channels are not memoryless, e.g. \cite{chen2015},
and commonly are made 
artificially so by interleaving for many existing
decoding schemes to function well, leading to additional delays in
encoding and decoding. In contrast, all of the results presented
in the present paper for GRAND and GRANDAB hold
directly for noise processes with more involved structures, and no
interleaving is required for their use. To illustrate that we have
presented analytic examples based on bursty Markovian noise.

The noise guessing approach underlying GRAND and GRANDAB has other
desirable features. For example, both schemes are universally
applicable in the sense that their execution only depends on the
structure of the noise rather than how the code-book was constructed.
Moreover, guesswork orders are known to be robust to mismatch
\cite{Sundaresan07}.

For both GRAND and GRANDAB, we provide asymptotic
results on the number of queries that the receiver must make per
received code-word for uniform-at-random code-books. Notably, the
approach becomes less complex as the code-book rate increases.
%potentially making it more applicable than prior methods. 

While testing a string's membership of a uniform-at-random code-book
can be achieved efficiently with the code-book stored in a $\A$-ary
tree, the code-book description requires substantial memory, limiting
utility for large block-lengths. 
Any use of a random code-book also requires  techniques for encoding, and for converting a code-word to an information word, but both of these can be performed with linear complexity. To encode, potential inputs can be stored in a lexicographically ordered $\A-$ary tree of depth $nR$ with a final leaf that contains a string of length $n$, the code-word to be transmitted. Thus finding an encoding entails a tree search, i.e. $nR$ operations. When mapping a code-word to an information word, the code-book can be stored as a lexicographically ordered $\A-$ary tree of depth $n$ with a final leaf that contains a string of length $nR$, the corresponding information-word. Thus identifying an information word also requires a tree search, i.e. $n$ operations.

An alternative instantiation of
the schemata would be realized in combination with linearly constructed
code-books such as Hamming, LDPC, or random linear code-books.
While
binary linear code-books can be capacity achieving for the BSC \cite{Gal68},
random linear code-books have recently been shown to be capacity-achieving
\cite{DZF16} for DMCs.
To describe a linear code-book, one need solely record
its generator matrix and so storage is small. Using the parity check
matrix associated with the generator, testing a string for membership
of the code-book is efficient as it only requires the computation
of the syndrome of the received string less guessed noise. Using
ML decoding by noise guessing with linear code-books effectively
results in replacing the usual coset leader of each syndrome, the
noise string in the coset with minimum Hamming weight, with the
most likely noise string in the coset. A thorough investigation of
that possibility, along with small block size properties, integration
into outer coding schemes, and so forth, is the topic of ongoing
work. The current work treats only a hard detection model where
only discrete data is presented to the decoder. Extending the
principles of these noise guessing techniques to a continuous case
where soft detection information is available imputes quantization
issues that merit their own investigation, and is the subject of
further ongoing work.

\section*{Acknowledgment}
This work is in part supported by the National Science Foundation under Grant No. 6932716.


\end{document}